\documentclass[showpacs,amsmath,amssymb,twocolumn,aps,pra,10pt,notitlepage,nofootinbib,superscriptaddress]{revtex4-2}

\usepackage[dvips]{graphicx} 
\usepackage{amsfonts}
\usepackage{amssymb}
\usepackage{amscd}
\usepackage{amsmath}    
\usepackage{amsthm}
\usepackage{bbm}
\usepackage{enumitem}
\usepackage{epsfig}
\usepackage{caption}
\usepackage{subcaption}
\usepackage[newcommands]{ragged2e}
\usepackage{xcolor}
\usepackage[all]{xy}
\usepackage{physics}
\usepackage{appendix}
\usepackage{algorithm}
\usepackage{algorithmicx}
\usepackage{algpseudocode}
\usepackage{soul}
\usepackage[inkscapelatex=false]{svg}
\usepackage{tikz}
\usetikzlibrary{calc, shapes.geometric, arrows, positioning}
\usepackage{comment}
\usepackage[colorlinks, linkcolor=red, anchorcolor=blue, citecolor=green]{hyperref}

\newtheorem{theorem}{Theorem}
\newtheorem{lemma}{Lemma}
\newtheorem{corollary}{Corollary}

\newtheorem{definition}{Definition}

\newcommand{\NCzero}{$\mathsf{NC}^0$}
\newcommand{\QNCzero}{$\mathsf{QNC}^0$}
\newcommand{\ClifNCzero}{$\mathsf{ClifNC}^0$}

\setlist[itemize]{parsep=0pt,leftmargin=*}
\setlist[enumerate]{parsep=0pt,leftmargin=*}

\begin{document}
\preprint{APS/123-QED}
\title{Unconditional quantum magic advantage in shallow circuit computation}
\date{\today}
\author{Xingjian Zhang}
\email{zxj24@hku.hk}
\affiliation{Center for Quantum Information, Institute for Interdisciplinary Information Sciences, Tsinghua University, Beijing 100084, China}
\affiliation{QICI Quantum Information and Computation Initiative, School of Computing and Data Science, University of Hong Kong, Pokfulam Road, Hong Kong}
\author{Zhaokai Pan}
\email{panzk24@mails.tsinghua.edu.cn}
\affiliation{Center for Quantum Information, Institute for Interdisciplinary Information Sciences, Tsinghua University, Beijing 100084, China}
\author{Guoding Liu}
\email{lgd22@mails.tsinghua.edu.cn}
\affiliation{Center for Quantum Information, Institute for Interdisciplinary Information Sciences, Tsinghua University, Beijing 100084, China}

\begin{abstract}
Quantum theory promises computational speed-ups over classical approaches. The celebrated Gottesman-Knill Theorem implies that the full power of quantum computation resides in the specific resource of ``magic'' states --- the secret sauce to establish universal quantum computation. However, it is still questionable whether magic indeed brings the believed quantum advantage, ridding unproven complexity assumptions or black-box oracles. In this work, we demonstrate the first unconditional magic advantage: a separation between the power of generic constant-depth or shallow quantum circuits and magic-free counterparts. For this purpose, we link the shallow circuit computation with the strongest form of quantum nonlocality --- quantum pseudo-telepathy, where distant non-communicating observers generate perfectly synchronous statistics. We prove quantum magic is indispensable for such correlated statistics in a specific nonlocal game inspired by the linear binary constraint system. Then, we translate generating quantum pseudo-telepathy into computational tasks, where magic is necessary for a shallow circuit to meet the target. As a by-product, we provide an efficient algorithm to solve a general linear binary constraint system over the Pauli group, in contrast to the broad undecidability in constraint systems. We anticipate our results will enlighten the final establishment of the unconditional advantage of universal quantum computation.

\end{abstract}

\maketitle

\section{Introduction}
Starting from Richard Feynman's proposal of simulating physics with quantum means~\cite{feynman1982simulating}, it has been an appealing quest to exploit phenomena unique to quantum theory to accelerate computation. A series of results, such as Shor's factoring algorithm~\cite{shor1997polynomialtime} and Grover's search~\cite{grover1997quantum}, strengthen the belief in the power of quantum computation. Notwithstanding the prosperity in the zoo of quantum algorithms, it is still intriguing to answer the basic questions: Does quantum theory really bring a computational advantage over classical means, and if yes, what is the origin of such power? A well-known statement that seems to respond to both questions is quantum ``magic''~\cite{bravyi2005universal}. The so-called quantum magic states are beyond the reach of stabilizer circuits, the ones initialized in the computational-basis state and composed of only Clifford gates and Pauli measurements. The Gottesman-Knill Theorem shows that stabilizer circuits can be perfectly simulated by classical computers in a polynomial time of the input size, deemed as efficient~\cite{gottesman1997stabilizer,aaronson2004improved,cormick2006classicality}. On the other hand, attempts from the simulation field suggest a strong relevance between the quantity of magic and the extent of quantum advantage~\cite{veitch2014resource,bravyi2016trading,howard2017application,seddon2019quantifying,bravyi2019simulationofquantum,wang2019quantifying,seddon2021quantifying,liu2022many,chen2023magic}. Quantum algorithms richer in magic are often more difficult for a classical computer to simulate.

Indeed, considering the structure of quantum state space, magic states, or equivalently non-Clifford operations, are indispensable for a complete picture~\cite{bravyi2005universal,zhu2016clifford}. In contrast, whether they bring a super-polynomial or even an exponential quantum computational advantage as promised remains to be proved. Despite numerous good reasons to believe in its validity~\cite{kitaev2002classical,preskill2023quantum}, unfortunately, explorations to date have not got rid of assumptions of unproven hardness for classical algorithms, such as factoring a large number in Shor's algorithm~\cite{shor1997polynomialtime}, or reliance on queries to a black-box oracle as in Grover's search~\cite{grover1997quantum}, where the oracle construction may be hard work.

To firmly establish the quantum advantage, one may alternatively start from a more restrictive regime in complexity. Instead of defining ``efficient'' as a polynomially growing time, a notable regime is the set of shallow circuits~\cite{cook1985taxonomy,hoyer2005quantum}, where the circuit depth, or equivalently the computing time, is restricted to a constant irrelevant to the problem size. The consideration of quantum shallow circuits was partly attributed to an experimental perspective, as it is relatively simpler to deal with system decoherence within a fixed time~\cite{chuang1995decoherence}. More importantly, theorists have rich toolkits from quantum information theory to aid the investigations. A particular instrument is quantum nonlocality, a most distinguishing property of quantum theory~\cite{brassard2005quantum,brunner2014bell}. As shown by the renowned Bell theorem~\cite{bell1964einstein}, entanglement leads to purely quantum correlations between nonlocal observers beyond the scope of classical physics~\cite{horodecki2009quantum}. One can translate quantum nonlocality into a computational task to generate nonlocal statistics among distant computing sites~\cite{bravyi2018quantum,bravyi2020shallow,caha2023colossal,bharti2023power}. While classical circuits require a growing time with respect to the input size to scramble the information for computation, quantum shallow circuits are competent to the task, bringing an unconditional advantage.

Despite the recent progress in shallow circuits, a vague question arises: Is magic indispensable for the full power of quantum computation in the low-complexity regime? Indeed, among all the existing explorations of quantum shallow circuits, the essential ingredient for the quantum advantage --- long-range entanglement, can be generated with Clifford circuits without using the magic resource~\cite{barrett2007modeling}. On the other hand, though not rigorous, with our experiences in the complexity theory such as the padding argument~\cite{arora2009computational}, we may be inclined to think of a collapse of the power of universal quantum computation if magic makes no difference in the low-complexity regime. Following the logic of relating nonlocality with shallow circuit computation, a relevant question is the role of quantum magic in nonlocality, which is much unexplored compared to the more prevalent quantum features like entanglement.

In this work, we unconditionally confirm that quantum magic brings an advantage, at least in a shallow circuit. For this purpose, we consider the strongest form of nonlocality, where nonlocal observers generate perfectly synchronous statistics, namely quantum ``pseudo-telepathy''~\cite{brassard2005quantum}. For the first time, we explicitly construct a quantum pseudo-telepathy correlation requiring magic resources and prove strict upper bounds on the correlation strength of solely magic-free operations. Then, we design a computational task that requires the output and the input to satisfy the constructed magic-necessary pseudo-telepathy correlations. This task separates the capabilities of generic quantum shallow circuits and their magic-free counterparts. We summarize our approach in Fig.~\ref{fig:Summary}.

\begin{figure*}[hbt!]
\centering 
\includegraphics[width=\textwidth]{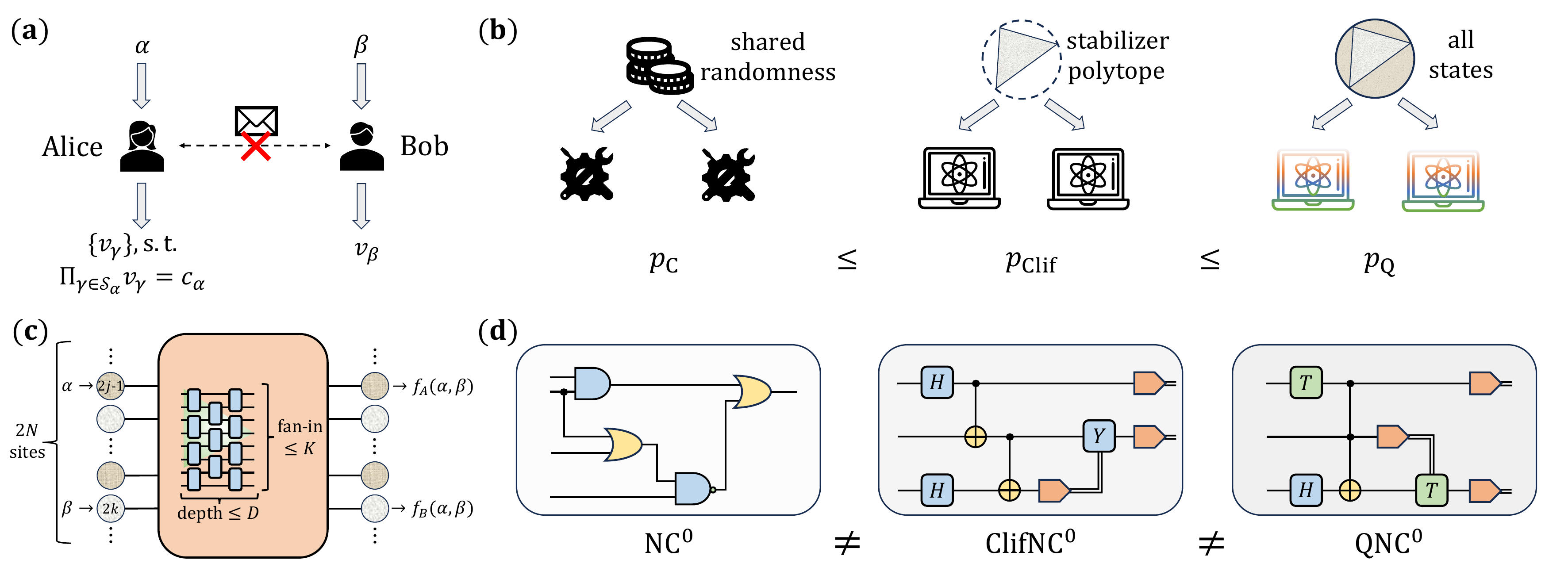}
\captionsetup{justification=Justified,singlelinecheck=false}
\caption{Summary of the results and methods. (a) A general BCS nonlocal game. Given a BCS, the question to Alice is a constraint indexed by $\alpha$, and the question to Bob is a variable indexed by $\beta$, with $\beta\in\mathcal{S}_\alpha$, where $\mathcal{S}_\alpha$ corresponds to the set of variables in constraint $\alpha$. The players win the game if and only if Alice outputs an assignment to the variables $\{v_\gamma\}_\gamma$ satisfying the constraint $\prod_{\gamma\in\mathcal{S}_\alpha}v_\gamma=c_\alpha$, and Alice and Bob give an identical assignment to $v_\beta$. (b) Different game strategies. (1) Classical strategy: players share randomness and apply local classical operations. (2) Clifford strategy: players share entanglement prepared by stabilizer circuits and apply local Clifford operations and local Pauli measurements. The set of stabilizer states forms a convex polytope~\cite{cormick2006classicality,veitch2014resource}. (3) General strategy: players share a general entangled state and apply general local quantum operations. The set of all quantum states is a convex set. There is a hierarchy among the maximum winning probabilities in each level: $p_{\mathrm{C}}\leq p_{\mathrm{Clif}}\leq p_{\mathrm{Q}}$. We present a family of nonlocal games via Eq.~\eqref{eq:PermutationBCS}, which manifest levelled quantum pseudo-telepathy, namely $p_{\mathrm{C}}<p_{\mathrm{Clif}}=1$ and $p_{\mathrm{Clif}}< p_{\mathrm{Q}}=1$ under different game parameters. (c) Solving a relation problem via a shallow circuit with bounded fan-in gates. The circuit depth $D$ and the maximum gate fan-in $K$ are both fixed constants. The BCS nonlocal game can be translated to a relation problem on $2N$ sites. See Sec.~\ref{sec:advantage} for the detailed construction. The bounded fan-in and constant depth conditions restrict each output site to be affected only by a constant number of input sites, as shown by the lightcone shaded in green. (d) Categories of shallow circuits. (1) \NCzero: classical shallow circuits with bounded fan-in gates, comprising classical gates like $\mathsf{AND},\mathsf{OR},\mathsf{NOT},\mathsf{NAND}$. (2) \ClifNCzero: magic-free shallow circuits with bounded fan-in gates, comprising a magic-free initial quantum state, Clifford gates like $\mathsf{H},\mathsf{S},\mathsf{CNOT}$, and Pauli measurements. (3) \QNCzero: general quantum shallow circuits with bounded fan-in gates, allowing a generic initial state and non-Clifford gates like the $\mathsf{T}$ gate and the Toffoli gate. Previous results have established the strict separation of $\mathsf{NC}^0\neq\mathsf{ClifNC}^0$~\cite{bravyi2018quantum,bravyi2020shallow}. In this work, we prove the strict separation of $\mathsf{ClifNC}^0\neq\mathsf{QNC}^0$.
}
\label{fig:Summary}
\end{figure*}

\section{Basic notions and main result}
The concept of quantum magic originates from the study of the classical simulation of quantum computation, closely related to the stabilizer formalism~\cite{gottesman1997stabilizer}. Consider an $n$-qubit quantum system on which the Pauli group is defined as the set of operators
\begin{equation}
  \mathbb{P}_n = \{\pm 1, \pm i\}\times \{\mathbb{I}, \sigma_x, \sigma_y, \sigma_z\}^{\otimes n},
\end{equation}
together with operator multiplication. Here, $\mathbb{I}$ is the two-dimensional identity operator, and $\sigma_x, \sigma_y, \sigma_z$ are the qubit Pauli matrices. The Clifford group is defined as the normalizer group of the Pauli group $\mathbb{P}_n$:
\begin{equation}
  \mathbb{C}_n = \{C\in \mathbb{U}_n|\forall P\in \mathbb{P}_n, CPC^{\dagger}\in \mathbb{P}_n\},
\end{equation}
where $\mathbb{U}_n$ is the set of all $n$-qubit unitary operators. Elements in the Clifford group are called Clifford operators or gates. A quantum circuit initialized in the computational-basis state and containing only Clifford gates and Pauli measurements is called a stabilizer circuit. Note that the Pauli measurements are equivalent to applying some Clifford gates, followed by computational-basis measurements. The Gottesman-Knill theorem states that the measurement results can be well-simulated by a classical circuit running in a time that is polynomial in the number of qubits~\cite{gottesman1997stabilizer,aaronson2004improved}. 

Clearly, there are unitary operators that do not belong to the Clifford group. Well-known non-Clifford operations include the $\mathsf{T}$-gate, which adds a non-trivial relative phase to basis state superposition:
\begin{equation}
    a\ket{0}+b\ket{1}\rightarrow a\ket{0}+e^{i\pi/4}b\ket{1},
\end{equation}
and the Toffoli gate, the quantum generalization of the $\mathsf{NAND}$ gate:
\begin{equation}
    \ket{c_1}\ket{c_2}\ket{t}\rightarrow\ket{c_1}\ket{c_2}\ket{t\oplus (c_1\cdot c_2)},
\end{equation}
where $c_1,c_2,t\in\{0,1\}$ represent the values in the two control qubits and the target qubit, respectively. Correspondingly, there are quantum states that cannot be prepared by any stabilizer circuit, even allowing post-selecting a subsystem upon Pauli measurement results. We call such states ``magic'' states. As an example, the state
\begin{equation}
    \ket{H}=\cos{\frac{\pi}{8}}\ket{0}+\sin{\frac{\pi}{8}}\ket{1}
\end{equation}
is a qubit magic state. 

To realize universal quantum computation and achieve quantum advantage, some sort of ``magic'' must be involved, which can be either some magic states or non-Clifford gates~\cite{raussendorf2001one,raussendorf2003measurement,bravyi2005universal,aaronson2004improved,bravyi2016trading}. In later discussions, we consider a model where all the magic comes from the gates. That is, the quantum state is initialized in $\ket{0}^{\otimes n}$, and the quantum measurement is performed on the computational basis. Depending on the type of computational resources, we categorize the circuits into three types. In a generic quantum circuit, the state undergoes operations in a universal gate set. Without strictly quantum gates, namely, all the quantum gates are restricted to a change of computational basis states, the circuit degenerates into a classical one. In between, we define the Clifford or magic-free circuit, in which the quantum gates must be Clifford.

Besides the accessible computational resources, the power of a circuit is also influenced by the circuit depth and the gate fan-in. The circuit depth is defined as the number of steps to perform all the gates and measurements. Note that in one step, multiple gates acting on different subsystems can be implemented in parallel. The gate fan-in is defined as the number of input (quantum) bits a gate can act on. For instance, the $\mathsf{T}$ gate has fan-in $1$, and the Toffoli gate has fan-in $3$. In our definition, the fan-in includes both classical bits and qubits, as depicted and explained in Fig.~\ref{fig:fanin}, e.g., a $\mathsf{T}$ gate controlled by a classical bit has fan-in 2. In this way, we unify the discussions for classical logical gates, quantum gates, classically controlled quantum gates, and measurements.

In this work, we are interested in the computational power of shallow circuits with bounded fan-in gates. That is, the circuit depth is a constant, and the fan-in of all the operations in the circuit has a constant upper bound. Within this restriction, we denote the classes of classical circuits~\cite{pippenger1979simultaneous,cook1985taxonomy}, Clifford circuits, and generic quantum circuits~\cite{hoyer2003quantum} as \NCzero, \ClifNCzero, and \QNCzero, respectively. We depict examples of these circuits in Fig.~\ref{fig:Summary}(d). As a remark, we overuse the complexity class notations for the associated circuits, like previous works in the field. 

Finally, we model the general description of a computational task, where a user of the circuit aims to compute a specific problem. For this purpose, the user interacts with the circuit by sending an input and asking it to return a desired output. Before reaching the final result, the user may interact with the circuit for multiple rounds, and the input in each round may depend on previous interactions. We can ask what the minimal circuit computational power is required to solve the problem.

Previously, a strict separation between \NCzero and \ClifNCzero was proved~\cite{bravyi2018quantum,bravyi2020shallow}. Here, we show magic further makes a fine structure among shallow circuits.
\begin{theorem}[Informal]
There is a separation between the circuit power:
    \begin{equation}\label{eq:main}
\text{\ClifNCzero}\neq\text{\QNCzero}.
\end{equation}
\end{theorem}

\begin{figure}[hbt!]
\centering
\includegraphics[width=0.56\linewidth]{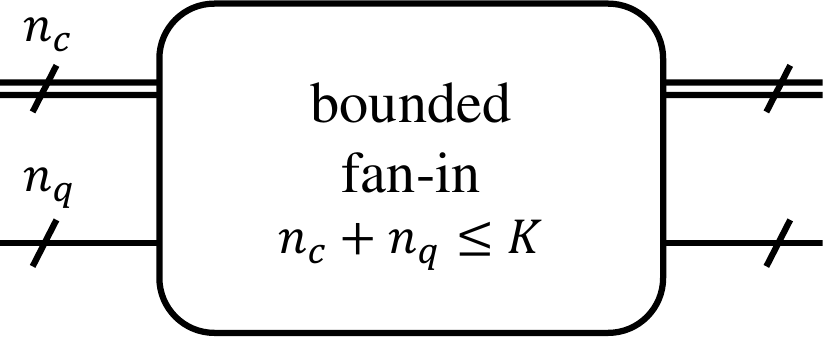}
\captionsetup{justification=Justified,singlelinecheck=false}
\caption{A $K$-bounded fan-in gate. In general, it acts on $n_c$ bits and $n_q$ qubits, with $n_c+n_q\leq K$. When $n_c=0$, it becomes a normal quantum gate characterized by a unitary operation. When $n_q=0$, it becomes a normal classical gate. }
\label{fig:fanin}
\end{figure}

\section{Binary-constraint-system-based nonlocal games}
To prove the main result in Eq.~\eqref{eq:main}, essentially, we manifest nonlocal correlations where magic states or non-Clifford operations play a non-trivial role. Our starting point is a special nonlocal game originating from the linear binary constraint system (BCS)~\cite{cleve2014characterization}. A BCS comprises a set of Boolean functions, namely constraints, over binary variables $v_\gamma$. We take the variable values over $\{+1,-1\}$ for later convenience. For a linear BCS, the constraints are given by functions in the form of $\prod_{\gamma\in\mathcal{S}_\alpha}v_\gamma=c_\alpha\in\{+1,-1\}$, where $\mathcal{S}_\alpha$ defines a subset of the variables. Given a linear BCS, consider a corresponding nonlocal game with two parties, Alice and Bob, as shown in Fig.~\ref{fig:Summary}(a). In each round of the game, a referee picks a constraint from the BCS labeled by $\alpha$ and a variable labeled by $\beta\in\mathcal{S}_\alpha$. The referee asks Alice to assign values to the variables satisfying the constraint and Bob to output a value for $v_\beta$. The nonlocal players win the game if and only if 
\begin{enumerate}
    \item Alice gives a satisfying assignment for the constraint $c_\alpha$, and
    \item Alice's assignment to $v_\beta$ coincides with Bob's. 
\end{enumerate}
Alice and Bob cannot communicate with each other once the game starts. Nevertheless, they can agree on a game strategy in advance. 

Naturally, the existence of a perfect winning strategy is related to the properties of the underlying linear BCS. If and only if the linear BCS has a solution where a fixed value assignment to the variables satisfies all the constraints, Alice and Bob can win the associated nonlocal game with certainty by classical means~\cite{cleve2014characterization}. Otherwise, the winning probability by any classical strategy, where the players are restricted to shared randomness and local classical operations, is strictly upper-bounded from $1$.

Notwithstanding, even if a fixed satisfying assignment does not exist, the nonlocal players may still win the game perfectly by exploiting quantum strategies. For a systematic study, we first generalize the BCS to a set of operator-valued functions. The scalar variables are replaced with Hermitian operators $A_\gamma$ of a finite dimension with eigenvalues $\{+1,-1\}$, and the constraints become $\prod_{\gamma\in\mathcal{S}_\alpha}A_\gamma=c_\alpha\mathbb{I}$ with $\mathbb{I}$ an identity operator. In addition, the observables corresponding to the variables in each constraint are required to be compatible, namely jointly measurable. The existence of a quantum perfect winning strategy is equivalent to the operator-valued BCS having a solution~\cite{cleve2014characterization}. Suppose the solution to the operator-valued BCS is given by a set of $d$-dimensional operators, $\{A_\gamma\}_\gamma$, then the perfect winning strategy in the corresponding nonlocal game goes as follows: 
\begin{enumerate}
    \item Alice and Bob first share a maximally entangled state, $\ket{\Phi^+}=\sum_{i=0}^{d-1}\ket{ii}/\sqrt{d}$;
    \item After the game starts, Alice measures the observables $\{A_\gamma\}_\gamma$, and Bob measures the observables $\{A_\gamma^{\mathrm{T}}\}_\gamma$ to assign values to the variables, where $\mathrm{T}$ denotes the operator transpose.
\end{enumerate}
By construction, Alice's measurement results satisfy the constraint. Also, as the maximally entangled state has the property
\begin{equation}
    \bra{\Phi^+}A_\gamma\otimes A_\gamma^{\mathrm{T}}\ket{\Phi^+}=\frac{1}{d}\tr(A_\gamma^2)=1,
\end{equation}
the assignments of Alice and Bob to the same variable thus coincide. 

Due to the intrinsic randomness in quantum measurements, an observable may take different outcomes in each constraint, hence assigning a different value to the same variable. 
Such flexibility brings an advantage over classical means, where quantum resources bring Alice and Bob ``pseudo-telepathy'' as if they knew what was going on at the other side via a ``spooky action''~\cite{brassard2005quantum}. We shall further discuss this issue and review relevant existing results in Appendix~\ref{supp:BCSPre}, such as the famous Mermin-Peres nonlocal game~\cite{mermin1990simple,peres1990incompatible}.

Among quantum strategies for the nonlocal game, there are also different levels of capabilities, as shown in Fig.~\ref{fig:Summary}(b). Instead of having access to all quantum states and operations, we consider constraining the players to applying only Clifford strategies. Specifically, Alice and Bob can apply Clifford operations to a state initialized in $\ket{0}$ before the nonlocal game to create entanglement. Afterward, they each take a share of the state and apply only Pauli-string measurements to the state for the game. If the nonlocal game has a Clifford strategy that wins perfectly, then the underlying BCS has a Pauli-string solution, and \emph{vice versa}. Notably, we derive the following results for general linear BCS. We prove Theorem~\ref{thm:PauliAlgor} in Appendix~\ref{algo} and Theorem~\ref{thm:PauliWinProb} in Appendix~\ref{appendsc:clifford}.

\begin{theorem}
    Given a linear BCS with $l$ variables and $m$ constraints, there exists a classical algorithm that finishes in $\mathrm{poly}(l,m)$ steps to determine whether the BCS has a Pauli-string operator-valued solution. If the answer is affirmative, the algorithm returns one such solution.
\label{thm:PauliAlgor}
\end{theorem}

\begin{theorem}
    Suppose a linear BCS does not have a Pauli-string solution. Then, for its associated nonlocal game, if Alice and Bob are restricted to Clifford strategies, either Alice fails to give satisfying assignments for all the constraints, or there exists one pair of questions $(\alpha,\beta)$, where the probability that Alice and Bob's assignments to $v_\beta$ coincide does not exceed $1/2$.
\label{thm:PauliWinProb}
\end{theorem}

Here, we describe the algorithm in Theorem~\ref{thm:PauliAlgor} for finding Pauli-string solutions to a general linear BCS, which highly relies on the following properties of Pauli operators.

\begin{lemma}
Suppose $A_1,A_2,\cdots,A_l$ are Pauli-string observables. For $j,k=1,\cdots ,l$, define $C_{jk}=A_jA_kA_jA_k$ as the commutator between $A_j$ and $A_k$. Then, $C_{jk}$'s have the following properties:
\begin{enumerate}
\item $C_{jk}\in\{\pm\mathbb{I}\}$. Specifically, $C_{jk}=\mathbb{I}$ when $A_jA_k-A_kA_j=0$, and $C_{jk}=-\mathbb{I}$ when $A_jA_k+A_kA_j=0$;
\item $C_{jk}=C_{kj}$ and $C_{jj}=\mathbb{I}$;
\item $A_jA_k=C_{jk}A_kA_j$.
\end{enumerate}
\end{lemma}

The proof of the lemma is straightforward. Now, suppose there exists a Pauli-string solution to the BCS with $l$ variables and $m$ constraints. Using this lemma, we can apply variable substitution and exchange the order between variables $A_j$ and $A_k$ similarly as solving a classical linear BCS, up to a sign change due to the Pauli operator commutation. In the end, we can express each variable $A_i$ in the BCS via a set of independent variables $\{A_r\}_r$ and sign variables $C_i$'s. In Appendix~\ref{algo}, we prove that with a further substitution of the expressions into the original BCS, we can transform the BCS into a linear BCS of solely the sign variables $C_i$'s and commutators $C_{jk}$'s between independent variables. The substitution thus far is efficient, namely in $\mathrm{poly}(l,m)$ steps. Since the new BCS is defined over $\mathbb{Z}_2$, it can be efficiently solved. If the new BCS does not have a solution, then by contradiction, the original BCS does not have a Pauli-string solution. 

If the new BCS has a solution, we can assign Pauli-string operators to the independent variables $\{A_r\}_r$ in the original BCS, which satisfies the required commutation conditions. Here, we give an explicit construction. Suppose there are $p$ commutators equal to $-1$, given by $C_{j_1k_1},C_{j_2k_2},\cdots,C_{j_pk_p}$. Then, we can construct Pauli strings over $p$ qubits according to the following rule: For every $q$'th qubit in each Pauli string, where $1\leq q\leq p$, assign $\sigma_x$ for $A_{j_q}$ and $\sigma_z$ for $A_{k_q}$; assign all the other qubits as $\mathbb{I}$. That is,
\begin{equation}
\begin{split}
\text{the $q$'th qubit of $A_r$}=\begin{cases}
\sigma_x, & \text{if $r=j_q$}, \\
\sigma_z, & \text{if $r=k_q$}, \\
\mathbb{I}, & \text{otherwise}.
\end{cases}
\end{split}
\end{equation}
It can be directly checked that this construction satisfies the requirements. The rest of the variables are then determined by the independent variables and sign variables $C_i$'s. This finishes the algorithm.

Theorem~\ref{thm:PauliAlgor} improves previous attempts of searching for a Pauli-string solution to a linear BCS~\cite{arkhipov2012extending,trandafir2022irreducible}, which pose additional requirements on the appeared times of each variable. Besides, this result sharply contrasts the common undecidability in the field of constraint systems, such as determining the existence of an operator-valued solution to a general BCS, which may not be Pauli strings~\cite{slofstra2019set}.

\section{Magic-necessary quantum pseudo-telepathy}\label{counterexample}
Previously, it was conjectured that whenever a linear BCS nonlocal game has a perfect winning strategy, it is either a Clifford or a classical one~\cite{arkhipov2012extending}. Recent group embedding results evidence the falseness of this conjecture~\cite{slofstra2020tsirelson,slofstra2019set}. Here, we take a relevant yet different approach and directly present a linear BCS nonlocal game to disprove it. To make it illustrative, we state the underlying BCS in the language of graph theory. Consider an undirected complete graph $G=(V,E)$ with $n$ vertices in the vertex set $V$. An undirected graph indicates that for any two connected vertices, $u$ and $v$, the tuples $(u,v)$ and $(v,u)$ represent the same edge in the edge set $E$. The BCS contains the following variables:
\begin{enumerate}
    \item Each vertex $v\in V$ corresponds to one variable $a_v$.
    \item Each undirected edge, denoted by $e=(u,v)\in E$, corresponds to three variables $x_{uv},y_{uv}$, and $z_{uv}$.
    \item Every two disjoint edges, denoted by $e_1=(u,v)\in E$ and $e_2=(s,t)\in E$, where $s,t,u$, and $v$ are different vertices, correspond to
    \begin{enumerate}
        \item one variable $b_{e_1e_2}\equiv b_{uv|st}$, where $b_{e_1e_2}=b_{e_2e_1}$;
        \item two variables $c_{e_1e_2}\equiv c_{uv|st}$ and $c_{e_2e_1}\equiv c_{st|uv}$, where $c_{e_1e_2}\neq c_{e_2e_1}$ in general.
    \end{enumerate}
\end{enumerate}
For clarity, we express the variables with respect to the underlying vertices and denote the BCS (nonlocal game) size with the number of vertices. Based on these variables, the BCS contains the following constraints:
\begin{equation}\label{eq:PermutationBCS}
\begin{split}
    a_u a_v y_{uv} &= 1, \forall (u,v)\in E, \\
    x_{uv} y_{uv} z_{uv} &= 1, \forall (u,v)\in E, \\
    x_{uv} x_{st} b_{uv|st} &= 1, \forall (u,v),(s,t)\in E, \\
    x_{uv} z_{st} c_{uv|st} &= 1, \forall (u,v),(s,t)\in E, \\
    b_{uv|st}b_{vs|ut}b_{su|vt} &=1, \forall u,v,s,t\in V, \\
    c_{uv|st}c_{vs|ut}c_{su|vt} &=1, \forall u,v,s,t\in V, \\
    \prod_{v\in V}a_v &= -1.
\end{split}
\end{equation}
The smallest non-trivial BCS that contains all types of variables is defined on a graph with four vertices, as shown in Fig.~\ref{fig:BCS}. BCS with a larger size can be defined similarly. Depending on the graph size, this family of BCS's exhibits a hierarchy among classical, Clifford, and general quantum resources, as shown by the following theorem.

\begin{figure}[hbt!]
    \centering
    \begin{tikzpicture}
        \node[circle, draw, inner sep=1pt,  label=130:$u$] (u) at (0.5, 1.6) {};
        \node[circle, draw, inner sep=1pt, label=190:$v$] (v) at (0, 0) {};
        \node[circle, draw, inner sep=1pt,  label=-20:$s$] (s) at (2.7, -0.1) {};
        \node[circle, draw, inner sep=1pt,  label=35:$t$] (t) at (2.5, 1.9) {};

        \definecolor{66ccff}{RGB}{102, 204, 255}
        \definecolor{ff9865}{RGB}{255, 152, 101}
        
        \draw[66ccff, very thick, dash pattern=on 6pt off 3pt] (u) -- (t);
        \draw[66ccff, very thick, dash pattern=on 6pt off 3pt] (v) -- (s);
        \draw[ff9865, very thick, densely dotted] (u) -- (s);
        \draw[ff9865, very thick, densely dotted] (v) -- (t);
        \draw[black, very thick] (u) -- (v);
        \draw[black, very thick] (s) -- (t);
    \end{tikzpicture}
\captionsetup{justification=Justified,singlelinecheck=false}
\caption{The undirected complete graph with four vertices. This graph defines the smallest non-trivial BCS variables in Eq.~\eqref{eq:PermutationBCS}. Each vertex $v$ corresponds to one variable $a_v$, and there are four such variables in the subgraph. Each undirected edge $e=(u,v)$ corresponds to three variables $x_{uv},y_{uv},z_{uv}$. With six edges in the subgraph, there are six variables for each kind. Every two disjoint edges $e_1=(u,v),e_2=(s,t)$ correspond to two kinds of variable $b_{uv|st}=b_{st|uv}$ and $c_{uv|st},c_{st|uv}$. The subgraph has three sets of disjoint edges, denoted by black solid lines, blue dashed lines, and orange dotted lines, respectively. Consequently, there are three variables of the kind $b_{uv|st}$ and six variables of the kind $c_{uv|st}$. }
\label{fig:BCS}
\end{figure}
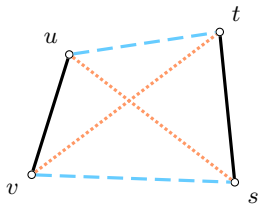

\begin{theorem}
    For the nonlocal game defined through the BCS in Eq.~\eqref{eq:PermutationBCS},
    \begin{enumerate}
        \item when $n=4$, it has a perfect-winning Clifford strategy, but it does not have a perfect-winning classical strategy;
        \item when $n\in2\mathbb{N}+5=\{5,7,9,\cdots\}$, it has a perfect-winning classical strategy;
        \item when $n\in 2\mathbb{N}+6=\{6,8,10,\cdots\}$, it has strategies that exploit quantum magic to win perfectly, but it does not have a perfect-winning Clifford strategy or classical strategy.
    \end{enumerate}
\label{thm:nonlocalgame}
\end{theorem}

In Appendix~\ref{supp:BCSResults}, we construct perfect-winning strategies in each level (Theorem~\ref{thm:classicalSol} for $n\in 2\mathbb{N}+5$, Theorem~\ref{thm:PauliSol} for $n=4$, and Theorems~\ref{thm:groupSol} and \ref{thm:CJrep} for $n\in 2\mathbb{N}+6$ in Appendix). The general algorithm in Theorem~\ref{thm:PauliAlgor} can aid the proof for the non-existence of Clifford strategies when $n\in 2\mathbb{N}+6$. Nevertheless, we provide a simpler proof tailored to this specific BCS nonlocal game (Theorem~\ref{thm:noPauliSol} in Appendix). As a corollary of Theorem~\ref{thm:PauliWinProb}, for the nonlocal game with $n\in 2\mathbb{N}+6$, with uniformly distributed random questions, the winning probabilities of all Clifford and classical strategies can be upper-bounded by
\begin{equation}\label{eq:CliWin}
    p_{\mathrm{Clif}}\leq1-\frac{1}{2|\mathcal{Q}|},
\end{equation}
where $\mathcal{Q}$ denotes the set of questions in the nonlocal game.

We present the calculation of $|\mathcal{Q}|$ below. For the BCS given in Eq.~\eqref{eq:PermutationBCS}, we denote the set of questions for Alice as $\tilde{\mathcal{Q}}^A$, namely the BCS constraints, and the set of questions for Bob as $\tilde{\mathcal{Q}}^B$, namely the BCS variables. In Table~\ref{table:NumConstraint} and~\ref{table:NumVariable}, we give the expressions to calculate the set sizes. We also provide the concrete numbers for the case of $n=8$. Note that the constraint $\prod_{v\in V}a_v=-1$ comprises $n$ variables, while every other constraint consists of three variables.

\begin{table}[hbt!]
\centering
\begin{tabular}{c|c|c}
\hline
\hline
constraint format & expression & number ($n=8$) \\
\hline
$aa'y=1$ & $\binom{n}{2}$ & 28 \\
$xyz=1$ & $\binom{n}{2}$ & 28 \\
$xx'b=1$ & $\binom{n}{4}\cdot 3$ & 210 \\
$xzc=1$ & $\binom{n}{4}\cdot 3\cdot 2$ & 420 \\
$bb'b''=1$ & $\binom{n}{4}$ & 70 \\
$cc'c''=1$ & $\binom{n}{4}\cdot 4$ & 280 \\
$\prod a=-1$ & 1 & 1 \\
total & $|\tilde{\mathcal{Q}}^A|$ & 1037 \\
\hline
\hline
\end{tabular}
\caption{Number of constraints in the BCS game.}
\label{table:NumConstraint}
\end{table}

\begin{table}[hbt!]
\centering
\begin{tabular}{c|c|c}
\hline
\hline
variable type & expression & number ($n=8$) \\
\hline
$a$ & $n$ & 8 \\
$x$ & $\binom{n}{2}$ & 28 \\
$y$ & $\binom{n}{2}$ & 28 \\
$z$ & $\binom{n}{2}$ & 28 \\
$b$ & $\binom{n}{4}\cdot 3$ & 210 \\
$c$ & $\binom{n}{4}\cdot 3\cdot 2$ & 420 \\
total & $|\tilde{\mathcal{Q}}^B|$ & 722 \\
\hline
\hline
\end{tabular}
\caption{Number of variables in the BCS game.}
\label{table:NumVariable}
\end{table}

For the convenience of the shallow circuit computational task, we slightly modify the BCS in Eq.~\eqref{eq:PermutationBCS}. For the $n$-variable constraint $\prod_{v\in V}a_v=-1$, we can introduce $(n-3)$ new variables and turn it into an equivalent set of $(n-2)$ constraints with three variables each. That is, we introduce new variables $a_{12},a_{123},\cdots,a_{1\cdots n-2}$ and convert the constraint as
\begin{equation}
    a_1a_2a_3\cdots a_n=-1
    \Longleftrightarrow
    \left\{\begin{array}{rl}
         a_1a_2a_{12}&=1 \\
         a_{12}a_3a_{123}&=1 \\
         \cdots \\
         a_{1\cdots n-3}a_{n-2}a_{1\cdots n-2}&=1 \\
         a_{1\cdots n-2}a_{n-1}a_n&=-1.
    \end{array}\right.
\end{equation}
Note that the commutation requirement between variables $a_u$ and $a_v$ in the original constraint is preserved, since they also need to satisfy the constraint of $a_ua_vy_{uv}=1$. Denote the set of constraints in the modified BCS as $\mathcal{Q}^A$. Suppose the original BCS with size $n$ consists of $|\tilde{\mathcal{Q}}^A|$ constraints. In correspondence with Eq.~\eqref{eq:CliWin}, there are $|\mathcal{Q}|=3|\mathcal{Q}^A|$ sets of questions in the modified BCS nonlocal game, with 
\begin{equation}
\begin{split}
|\mathcal{Q}^A|&= |\tilde{\mathcal{Q}}^A|+n-3 \\ 
&=2\binom{n}{2}+14\binom{n}{4}+n-2.\\
\end{split}
\end{equation}
Therefore, we obtain a direct upper bound on the average winning probability of Clifford strategies as $1-1/6|{\mathcal{Q}}^A|$.

In the construction of perfect-winning strategies for BCS games with size $n\in2\mathbb{N}+6$, we apply a group-theoretic method~\cite{cleve2017perfect,coladangelo2017robust}. A notable property is that there are non-unique solutions to the BCS and thus non-equivalent perfect winning quantum strategies in this case. The measurements in different strategies cannot be transformed into each other via a local isometry or the complex conjugate operation, hence the BCS game does not manifest a self-testing property~\cite{supic2020self}. Moreover, as the strategies take different dimensions, the associated maximally entangled states are thus non-equivalent~\cite{paddock2024operator}. This property even distinguishes our BCS game from nonlocal games with the ``weak-form'' self-testing property~\cite{kaniewski2020weak}, where while the measurements can be non-unique in the optimal quantum strategies, they require the same entangled state with a fixed dimension up to a local isometry for the optimal quantum strategy. 

Here, we present one operator-valued solution to the BCS when $n=8$. Labelling the vertices from $1$ to $8$, a realization of $a_v$ and $x_{uv}$ in the above BCS is
\begin{equation}
\begin{split}
a_v &=\mathbb{I}_8-2\mathbf{e}_{vv},v=1,\cdots,8, \\
x_{uv} &=\mathbb{I}_8-\mathbf{e}_{uu}-\mathbf{e}_{vv}+\mathbf{e}_{uv}+\mathbf{e}_{vu},u,v=1,\cdots,8, u\neq v,
\end{split}
\end{equation}
where $\mathbb{I}_8$ is an eight-dimensional identity operator, and $\mathbf{e}_{ij}$ denotes an elementary matrix, of which the element in the $i$'th row and $j$'th column is one, and all the other elements are zero. The other operators can be determined via $a_v$'s and $x_{uv}$'s.
The perfect winning strategy of the nonlocal game thus takes three pairs of the Einstein-Podolsky-Rosen (EPR) state, $(\ket{00}+\ket{11})/\sqrt{2}$. In this strategy, the measurements require non-Clifford operations. As an example, Fig.~\ref{fig:realization} depicts the observables of $x_{78},y_{78},z_{78}$ and the implementation for the constraint $x_{78}y_{78}z_{78}=1$. The observable $x_{78}$ is the same as the Toffoli gate. To measure this observable, one can use the Hadamard test, in which one applies the non-Clifford gate of controlled-controlled-controlled-X (CCCX) followed by the computational basis measurement. Similarly, the non-Clifford gates required to measure $y_{78}$ and $z_{78}$ are the controlled-controlled-Z (CCZ) and controlled-controlled-controlled-(-X) gates, respectively.

\begin{figure}[hbt!]
\centering 
\includegraphics[width=0.32\textwidth]{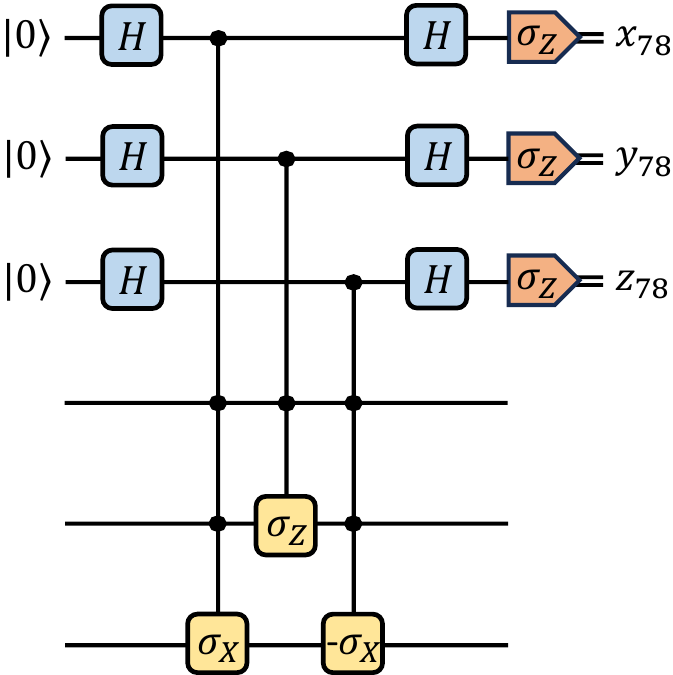}
\captionsetup{justification=Justified,singlelinecheck=false}
\caption{The circuit with non-Clifford gates to realize the simultaneous measurement of $x_{78}$, $y_{78}$, and $z_{78}$ by using the Hadamard test. The Hadamard test uses an ancilla initialized in $\ket{0}$, followed by a Hadamard gate, applying controlled-$O$ and again a Hadamard gate, and measuring in the computational basis to get measurement results of $O$. Here, the non-Clifford gates are CCCX, CCZ, and CCC(-X) for $x_{78}$, $y_{78}$, and $z_{78}$, respectively.}
\label{fig:realization}
\end{figure}

\section{Magic computational advantage in shallow circuits}\label{sec:advantage}
As shown in Theorem~\ref{thm:nonlocalgame}, the nonlocal game defined through the BCS in Eq.~\eqref{eq:PermutationBCS} separates the capabilities between a generic quantum world and the magic-free world to generate correlations. Now, we translate correlation generation into a computational task of a relation problem and show the magic advantage. As a reminder, the computational task is a single-user one. ``Alice'' and ``Bob'' in the nonlocal game now refer to parts of the circuit, which is merely for intuitive thinking. In particular, one should not consider the task as a distributed computation. 

Briefly speaking, a relation problem randomly selects an input bit string $z_{\mathrm{in}}$ from a set and asks the computation to output a bit string $z_{\mathrm{out}}$, such that $z_{\mathrm{out}}$ always satisfies a certain relation with respect to $z_{\mathrm{in}}$. Given a nonlocal game with size $n$, we can define a relation problem labeled by $n$, $R_N^n$. As shown in Fig.~\ref{fig:Summary}(c), one can imagine that two experimentalists, Alice and Bob, each holds $N$ computing sites and collaborate to solve $R_N^n$. We use the capital letter $N$ for the number of computing sites to distinguish it from the underlying nonlocal game size $n$. The problem of $R_N^n$ randomly specifies two sites of Alice and Bob, denoted as $2j-1$ and $2k$, respectively, and inputs questions in the nonlocal game, which we encode as bit strings $\alpha\in\mathcal{Q}^A,\beta\in\mathcal{Q}^B$. Other sites are input with a fixed value, $\perp$. The sites $2j-1$ and $2k$ are required to output $f_A(\alpha,\beta)$ and $f_B(\alpha,\beta)$, respectively.

In a circuit comprising $K$-bounded fan-in gates, where each gate can act on at most $K$ inputs, the value $K$ mimics the light speed for information scrambling~\cite{lieb1972velocity}. Furthermore, if the circuit is shallow, where the circuit depth is a constant independent of the problem size, it restricts the ``time'' for information scrambling; hence, many sites in the circuit are ``space-like'' separated from each other.

When applying a shallow circuit with bounded fan-in gates to solve the relation problem defined by a nonlocal game, Alice and Bob must be capable of winning the game between space-like separated sites without communication. Suppose the underlying nonlocal game cannot be won perfectly without a particular resource, which is quantum magic in our discussion. In that case, the players must communicate to exchange information and generate the desired correlation, which takes time. Therefore, the shallow circuit should fail in the task. 

On the contrary, suppose this particular resource exists such that the nonlocal game can be won perfectly. Entanglement can be created and distributed between two arbitrary sites via entanglement swapping~\cite{pan1998experimental} with bounded fan-in quantum gates in constant steps~\cite{bravyi2020shallow,bharti2023power}, and quantum pseudo-telepathy completes the remaining task. A caveat is that although the distributed state after entanglement swapping $\ket{\psi}$ is still maximally entangled, conditioned on the quantum measurement result in this procedure, the state might undergo a random local Pauli error, deviating from the EPR state $\ket{\Phi^+}$ necessary for perfect winning the BCS nonlocal game. Due to the non-Clifford nature, the measurements of Alice and Bob are not perfectly correlated, namely $\bra{\psi}A\otimes A^{\mathrm{T}}\ket{\psi}\neq1$. They are neither perfectly anti-correlated, namely $\bra{\psi}A\otimes A^{\mathrm{T}}\ket{\psi}\neq-1$, such that we can embed the correlation flip to the defining relation~\cite{bravyi2020shallow}. Without communication to correct the Pauli error based on the error syndrome, these cases render in random statistics that violate the relation. To solve this issue, we consider a two-round interaction. In the first round, the user simply picks two computational sites at random, and the shallow circuit carries out entanglement swapping to distribute the EPR state to them. In the second round, besides asking nonlocal game questions to these sites, the user also sends an error syndrome to one of them, giving it the necessary information to correct the Pauli error. By construction, a generic quantum shallow circuit that embeds the nonlocal game strategy as a sub-routine is competent to the task.

Building on the nonlocal games defined through Eq.~\eqref{eq:PermutationBCS}, the following theorem illustrates the separation between \QNCzero and \ClifNCzero circuits in solving the BCS relation computational problem.

\begin{theorem}
Given the nonlocal game defined by the BCS in Eq.~\eqref{eq:PermutationBCS} with size $n\in 2\mathbb{N}+6$, there is a relation problem $R_N^n$ carried out in two rounds and a constant $K_{\mathrm{th}}$ independent of $N$, such that for any integer $K>K_{\mathrm{th}}$,
\begin{itemize}
\item
$R_N^n$ can be perfectly solved by a fixed \QNCzero circuit with $K$-bounded fan-in one-dimensional geometrically local gates, where some gates are non-Clifford operations.

\item
Any probabilistic Clifford circuit with $K$-bounded fan-in gates that solves $R_N^n$ with probability larger than $(1+p_{\mathrm{Clif}})/2$, where $p_{\mathrm{Clif}}$ is given in Eq.~\eqref{eq:CliWin}, must have a circuit depth at least increasing logarithmically with respect to the problem size $N$, where the gates can be non-geometrically local.
\end{itemize}
\label{thm:relation}
\end{theorem}

To illustrate the relation problem in more detail, in the following discussions, we focus on $R^8_N$, namely the problem that embeds the BCS nonlocal game with size $n=8$. Other values of $n$ can be studied similarly. We describe the task by giving the quantum shallow circuit that perfectly solves the computational problem. As shown in Fig.~\ref{fig:relation}(b), consider a quantum circuit containing $2N$ computing sites labeling from $1$ to $2N$. Imagine the circuit is divided into two parts, Alice and Bob, where Alice holds the odd-valued sites, and Bob holds the even-valued sites. Each site consists of a set of classical wires to receive the input of $R^8_N$ and quantum wires initialized in $\ket{0}$ for three qubits. Throughout this work, we denote the classical systems by double wires and quantum systems by single wires.

In the first round of the computation, the user randomly picks up two computational sites $(j,k)$, as depicted by the orange dots, on which they require the shallow circuit to distribute entanglement.
For this purpose, Alice and Bob first prepare three pairs of the EPR state, $(\ket{00}+\ket{11})/\sqrt{2}$, between their adjacent sites $(2i-1,2i),i\in\mathbb{N}^+$, as shown by the blue boxes. Next, Alice and Bob perform three Bell-state measurements (BSMs) between the three pairs of qubits on their adjacent sites $(2j,2j+1),\cdots,(2k-2,2k-1)$ in a concatenated manner, as shown by the white boxes. A BSM projects two qubits into one of the four orthogonal Bell states,
\begin{equation}
\begin{split}
    \ket{\Phi^+}&=\frac{\ket{00}+\ket{11}}{\sqrt{2}}, \\
    \ket{\Phi^-}&=\frac{\ket{00}-\ket{11}}{\sqrt{2}}, \\
    \ket{\Psi^+}&=\frac{\ket{01}+\ket{10}}{\sqrt{2}}, \\
    \ket{\Psi^-}&=\frac{\ket{01}-\ket{10}}{\sqrt{2}}, \\
\end{split}
\end{equation}
and the measurement can be realized by a $\mathsf{CNOT}$ gate with Pauli measurements. When the BSM is performed jointly on one qubit of two Bell states each, the remaining two qubits become one of the Bell states with equal probability conditioned on the measurement outcome, namely an entanglement swapping. Afterward, the user reads the BSM results and calculates the error syndrome of the distributed entangled state.

\begin{figure*}[hbt!]
\centering
\begin{subfigure}[b]{0.46\textwidth}
    \centering
    \includegraphics[width=\linewidth]{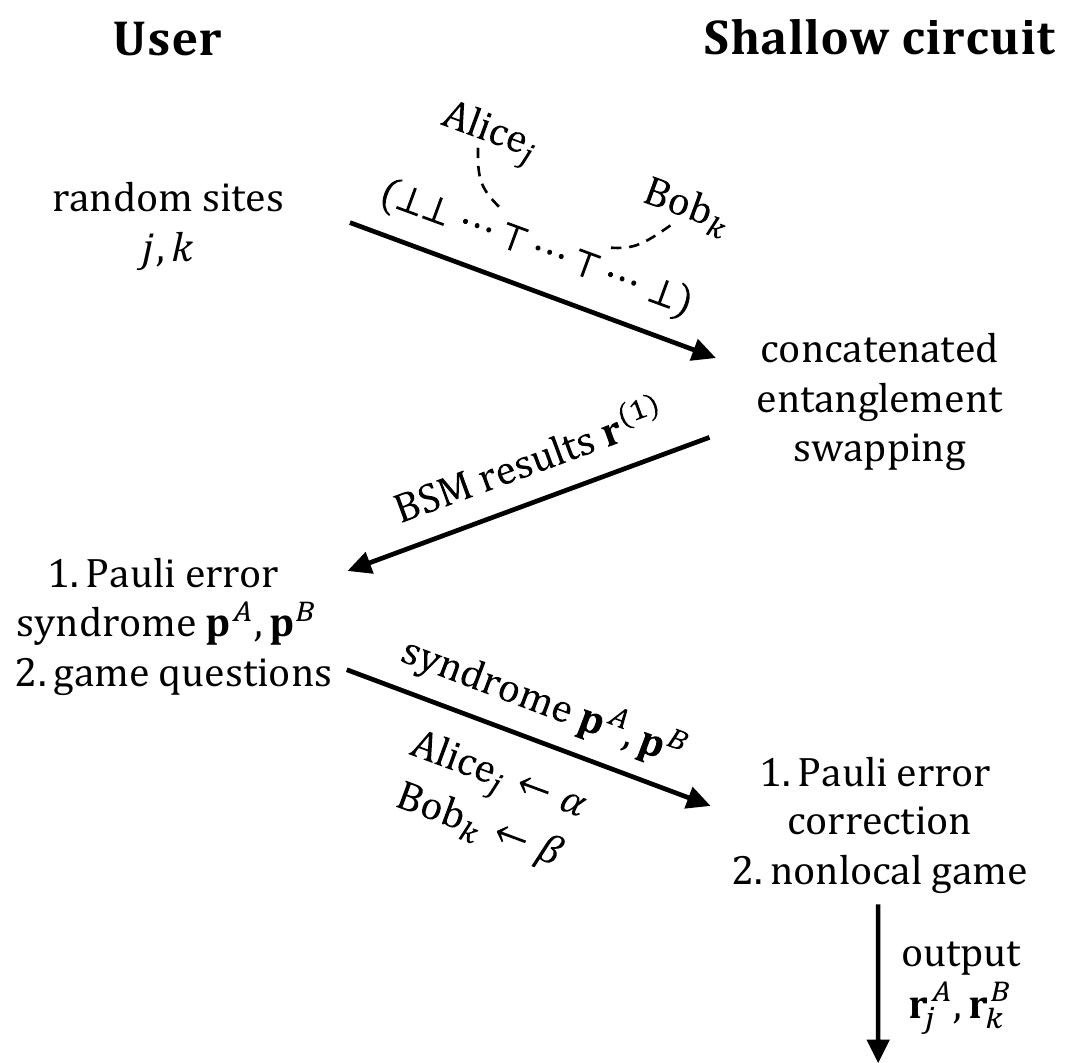}
    \caption{}
\end{subfigure}
\begin{subfigure}[b]{0.52\textwidth}
    \centering
    \includegraphics[width=\linewidth]{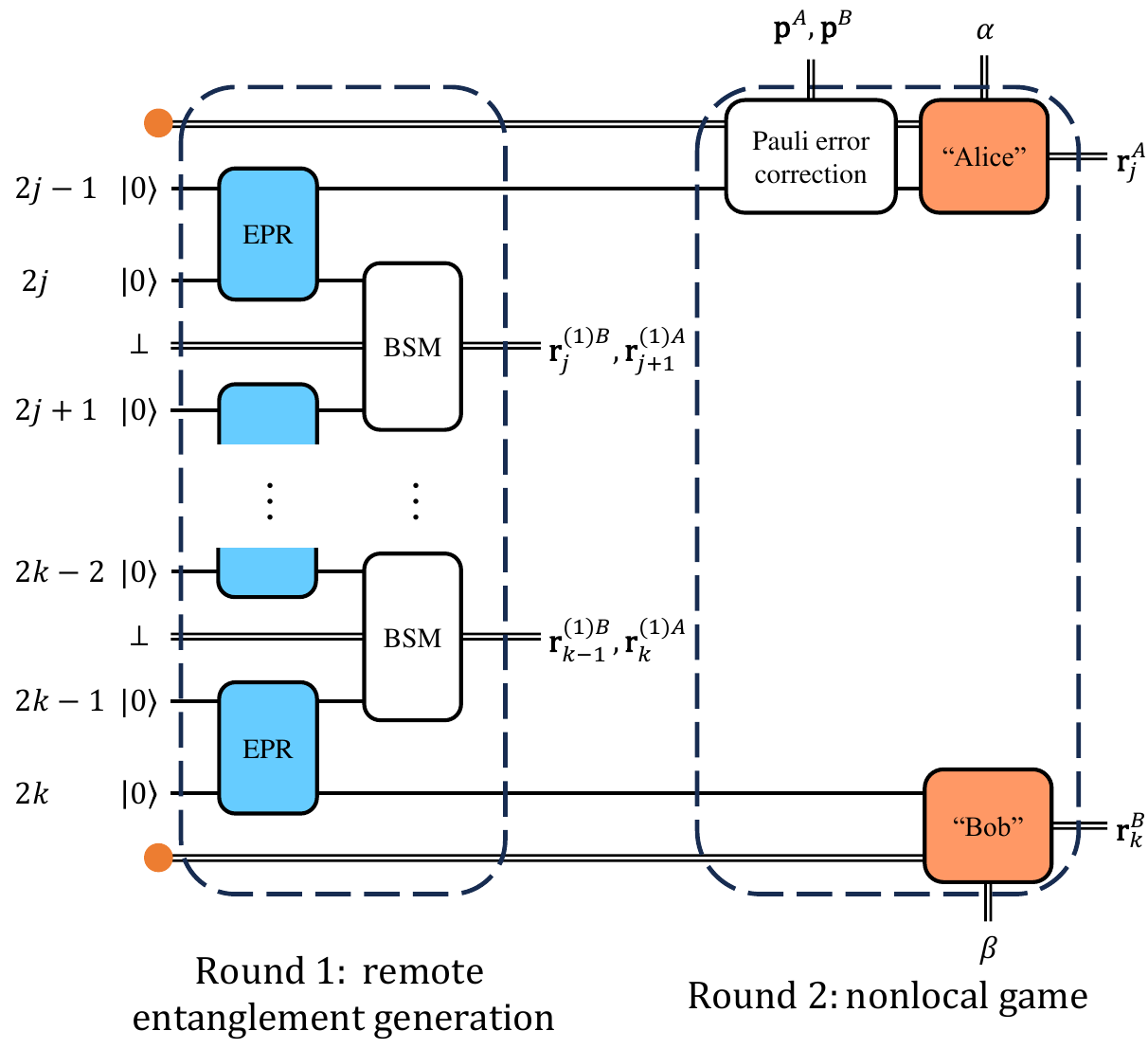}
    \caption{}
\end{subfigure}
\hfill
\captionsetup{justification=Justified,singlelinecheck=false}
\caption{Description of the interactive relation problem. Schematic diagram (a) of the interactions in the relation problem $R^n_N$ and a shallow circuit construction (b) to solve it. In the first round, the user randomly picks computation sites $j,k$, on which the nonlocal game shall be taken in the second round. The user informs the circuit of this choice with the character $\top$ on the positions $\mathrm{Alice}_j$ and $\mathrm{Bob}_k$ in a $2N$-bit string. The other positions in the bit string take the character $\bot$.
In the quantum shallow circuit for this round [the left dashed box in (b)], EPR states are initially prepared in each pair of sites between Alice and Bob (shown by the blue boxes). Upon receiving the input in this round, the circuit performs BSMs between adjacent EPR states (shown by the white boxes) except for the positions receiving $\top$ (shown by the orange dots). The measurement outcomes $\mathbf{r}^{(1)}$ are taken as the outputs of the first round, which is sent back to the user. Computation sites $\mathrm{Alice}_j$ and $\mathrm{Bob}_k$ shall share EPR states up to a local Pauli flip related to the measurement outcomes. 
In the second round, the user calculates the Pauli error syndrome $\mathbf{p}^A,\mathbf{p}^B$ from $\mathbf{r}^{(1)}$ by Eq.~\eqref{eq:PauliErrorSyndrome} and randomly picks questions $\alpha\in\mathcal{Q}^A,\beta\in\mathcal{Q}^B$ in the BCS nonlocal game defined by Eq.~\eqref{eq:PermutationBCS}. Afterwards, the user sends the syndrome $\mathbf{p}^A,\mathbf{p}^B$ and the question $\alpha$ to $\mathrm{Alice}_j$ and sends the question $\beta$ to $\mathrm{Bob}_k$. In the quantum shallow circuit for this round [dashed box on the right of (b)], $\mathrm{Alice}_j$ first corrects the Pauli error according to $\mathbf{p}^A,\mathbf{p}^B$ (shown by the white box) and then performs the perfect-winning strategy for question $\alpha$ (shown by the orange box denoted as ``Alice''); $\mathrm{Bob}_k$ performs the perfect-winning strategy for question $\beta$ (shown by the orange box denoted as ``Bob''). The final computation output is the measurement results $\mathbf{r}_j^A,\mathbf{r}_k^B$ in the nonlocal game.
}
\label{fig:relation}
\end{figure*}


In the second round, besides sending the randomly chosen nonlocal game questions to the sites, the user also sends the error syndrome to one of the sites, say $j$, letting the circuit first correct the Pauli error before playing the nonlocal game. We depict the interactions between the user and the shallow circuit in Fig.~\ref{fig:relation}(a) and a circuit competent to the task in (b). Note that each site performs either the BSM or the nonlocal game operations conditioned on the input information, and we omit the quantum gates that do not take actions in the figure. The rigorous description of the task is given as follows.

\begin{enumerate}
    \item[] \textbf{Relation problem $R^n_N$:}
    
    Given a BCS game defined by Eq.~\eqref{eq:PermutationBCS} of size $n$, a problem instance of $(j,k,\alpha,\beta)$ where $1\leq j<k\leq N,\alpha\in\mathcal{Q}^A,\beta\in\mathcal{Q}^B$, defining:
    \item Round 1:
    \begin{itemize}[label={}]
        \item \textbf{Input:}
    \begin{equation}\label{eq:round1in}
    \begin{split}
        &\alpha_i=\begin{cases}
        \top, & \text{if $i=j$}, \\
        \bot, & \text{if $i\neq j$},
        \end{cases} \\
        &\beta_i=\begin{cases}
        \top, & \text{if $i=k$}, \\
        \bot, & \text{if $i\neq k$}.
        \end{cases}
    \end{split}
    \end{equation}

    \item \textbf{Valid outputs:} \\
    Any possible bit strings $\{r^{(1)A}_i(l)\}_i$ and $\{r^{(1)B}_i(l)\}_i,l\in\{1,2,3\}$ where all bits take values in $\{\pm 1\}$.
    \end{itemize}
    
    \item Round 2:
    \begin{itemize}[label={}]
        \item \textbf{Input:}
    \begin{equation}
        \begin{split}
            &\alpha_i=\begin{cases}
        \left(\mathbf{p}^A, \mathbf{p}^B, \alpha\right), & \text{if $i=j$}, \\
        \left(\mathbf{1}, \bot\right), & \text{if $i\neq j$},
        \end{cases} \\
        &\beta_i=\begin{cases}
        \beta, & \text{if $i=k$}, \\
        \bot, & \text{if $i\neq k$},
        \end{cases}
        \end{split}
    \end{equation}
    where 
    \begin{equation}\label{eq:PauliErrorSyndrome}
    \begin{split}
        \mathbf{p}^A(l) &= \prod_{i=j+1}^{k}r^{(1)A}_i(l), \\
        \mathbf{p}^B(l) &= \prod_{i=j}^{k-1}r^{(1)B}_i(l),
        \quad l\in\{1,2,3\},
    \end{split}
    \end{equation}
    and $\mathbf{1}$ represents the vector with all elements $1$.

    \item \textbf{Valid outputs:} \\
    Bit strings $\{r^{A}_i(l)\}_i$ and $\{r^{B}_i(l)\}_i,l\in\{1,2,3\}$ satisfying
    \begin{align}
        (\mathbf{r}_j^A,\mathbf{r}_k^B)&=f(\alpha,\beta),
    \end{align}
    where $f(\alpha,\beta)$ is the satisfying assignment in the nonlocal game.
    \end{itemize}
\end{enumerate}

Finally, we make some remarks on the computational task we construct. For the circuit to solve the computational problem, we require its configuration of gates and internal states to be fixed before the task starts. The circuit configuration must not change between the rounds. Nevertheless, the circuit may utilize randomness independent of the input for the computation. To overcome the problem of Pauli errors ruining the correlation of non-Clifford operator measurements, we introduce an additional interaction, where the user sends syndrome information to the shallow circuit to correct the Pauli error in entanglement distribution. From a physical viewpoint, the interactions are similar to the operations in a Bell test experiment based on a heralding entanglement source~\cite{hensen2015loophole}. The syndrome input in the second-round computation mimics the heralding signal in the Bell test, which indicates whether the desired entangled state is distributed before the nonlocal game starts.

Details about the lower bound of Clifford circuit depth are given in Appendix~\ref{supp:interactive}. As a remark, the logarithmic separation in the second part of the theorem is tight. That is, a magic-free circuit with bounded fan-in gates can solve the problem, of which the depth grows logarithmically. For a straightforward solution, all the computing sites send their input to a fixed ancilla, which performs the nonlocal game calculation and sends back the result to the specified sites. By this, we prove the main result in Eq.~\eqref{eq:main}.

Besides considering a two-round protocol to correct the Pauli error manually, another way to tackle the issue is to discard the case of Pauli errors. That is, we modify the original relation problem and define a one-round sampling problem, which requires Alice and Bob to play the nonlocal game when entanglement swapping has no Pauli error and allows arbitrary outputs when a Pauli error exists. We also require the first case to appear with a non-negligible probability. If no additional randomness is allowed, any fixed \ClifNCzero circuit fails to solve this sampling problem. Details of the sampling problem are given in Appendix~\ref{app:sampling}.


\section{Discussion}
In summary, we discover a family of nonlocal games that require quantum magic to win perfectly and translate it into an unconditional proof of magic computational advantage. While not being the purpose of this work, by applying the ``game gluing'' technique~\cite{ji2013binary,coladangelo2017robust}, one can combine games with different sizes in the family and construct a nonlocal game with a strict separation between the winning probabilities of classical, Clifford, and general quantum strategies. We consider this to be of independent interest to some research. In addition, we believe other nonlocal games exist that can demonstrate magic advantage. For instance, following the method of embedding a general group into a BCS in Ref.~\cite{slofstra2019set}, one can obtain candidate BCS nonlocal games. In Appendix~\ref{Supp:Slofstra}, we review the procedure. Combined with Theorem~\ref{thm:PauliAlgor}, one can efficiently check whether they have perfect-winning Clifford strategies.

To compute the relation problem in a realistic experiment, one shall further consider the noise. For this purpose, noise-tolerant methods, such as error correction and mitigation in a shallow circuit, need to be developed. With recent experimental progress, preparing magic states and implementing a few layers of non-Clifford gates are becoming easy~\cite{arute2019supremacy,wu2021advantage}. We expect that the computational task in this work can be faithfully realized on an upcoming early fault-tolerant quantum computing platform~\cite{sivak2023real,ni2023beating,bluvstein2024logical}.

This work takes the first step in proving the computation necessity of quantum magic unconditioned on any complexity assumption. We hope our results can inspire further explorations in this direction, eventually going beyond the regime of shallow circuits and solidifying the ``magic'' of universal quantum computation.

\section*{Acknowledgements}
We acknowledge Zhengfeng Ji and Honghao Fu for the insightful discussions on the binary constraint systems and the group embedding results, Qi Zhao, Zhaohui Wei, and Yilei Chen for leading us to the consideration of non-Clifford quantum operations, Yuwei Zhu, Boyang Chen, and Yuxuan Yan for helpful discussions on the group representation theory and quantum magic, Yu Cai for discussions on quantum self-testing, and Yuming Zhao for explaining the results in Ref.~\cite{paddock2024operator}. 
We express special thanks to Tian Ye for his vital and generous help during the early stages of this project in a weekly discussion. 
This work was supported by funding from the National Natural Science Foundation of China Grant No.~12174216 (XZ, ZP, GL) and the Innovation Program for Quantum Science and Technology Grant No.~2021ZD0300804 (XZ, ZP, GL). 
XZ acknowledges additional support by the Hong Kong Research Grant Council (RGC) through grant number R7035-21 of the Research Impact Fund and No. 27300823, N\_HKU718/23, HKU Seed Fund for Basic Research for New Staff via Project 2201100596, Guangdong Natural Science Fund via Project 2023A1515012185, National Natural Science Foundation of China (NSFC) via Project No. 12305030 and No. 12347104, and Guangdong Provincial Quantum Science Strategic Initiative GDZX2200001.

\section*{Author Contributions}
X.Z., Z.P., and G.L. contributed equally to this work in all stages, including initializing ideas, deriving proofs, and paper writing.

\clearpage

\appendix
\onecolumngrid


\section{Preliminaries}\label{Supp:Pre}
In this section, we review preliminary concepts for this work. We assume readers are familiar with the basic notions of linear algebra, graph theory, and group theory, and the basic description of quantum systems. For completeness, we restate some basic notions that are listed in the main text.

In the Appendix, we overuse some letters such as $n$ and $i$ when expressing the total number of items or labelling the variables; nevertheless, their meaning can be specified from the context.

\subsection{Quantum magic and non-Clifford operations}
We first briefly review the concept of quantum magic and related notions. For simplicity, we only consider the $n$-qubit system based on the Pauli group. Nevertheless, the results can be easily generalized to systems with a prime dimension by using the Weyl-Heisenberg algebra. Readers who are interested in the topic may refer to the Ph.D. thesis of Gottesman~\cite{gottesman1997stabilizer} and its following works for a more in-depth discussion.

Let us start with the definition of the Pauli observables on a single qubit:
\begin{equation}
\mathbb{I} = \begin{pmatrix}
1 & 0\\
0 & 1\\
\end{pmatrix},
\quad
\sigma_x = \begin{pmatrix}
0 & 1\\
1 & 0\\
\end{pmatrix},
\quad
\sigma_y = \begin{pmatrix}
0 & -i\\
i & 0\\
\end{pmatrix},
\quad
\sigma_z = \begin{pmatrix}
1 & 0\\
0 & -1\\
\end{pmatrix},
\end{equation}
where $\mathbb{I}$ is the identity operator, and the other three observables $\sigma_x$, $\sigma_y$, and $\sigma_z$ are always named nontrivial Pauli observables. The $n$-qubit Pauli group is defined as the set of operators
\begin{equation}
\mathbb{P}_n = \{\pm 1, \pm i\}\times \{\mathbb{I}, \sigma_x, \sigma_y, \sigma_z\}^{\otimes n},
\end{equation}
together with the operator multiplication. The Clifford group is defined as the normalizer of the Pauli group $\mathbb{P}_n$:
\begin{equation}
\mathbb{C}_n = \{C\in \mathbb{U}_n|\forall P\in \mathbb{P}_n, CPC^{\dagger}\in \mathbb{P}_n\},
\end{equation}
where $\mathbb{U}_n$ is the $n$-qubit unitary group. Operators in the Clifford group are called Clifford operations or gates. A highly related concept is the stabilizer state, which is generated by applying Clifford gates on the computational basis states, or equivalently the eigenstates of $\sigma_z^{\otimes n}$. If a state cannot be prepared in this way or by mixing stabilizer states, the state is said to contain quantum ``magic''~\cite{bravyi2005universal}.

\subsection{General binary constraint systems}\label{supp:BCSPre}
In this section, we review the definition of a binary constraint system (BCS). 
A BCS consists of $n$ binary variables $v_1,\cdots,v_n$ and $m$ constraints $c_1,\cdots,c_m$, where each $c_j$ is a Boolean equation with respect to a subset of $v_i$'s, $v_i\in\mathcal{S}_j$. In later discussions, we shall specify a constraint by $c_j$. Note that BCS with general Boolean constrains can describe general systems of equations~\cite{ji2013binary}. In a linear BCS, all the constraints are given by addition over $\mathbb{Z}_2$, or the parity operation over Boolean variables ranging in $\mathbb{Z}_2=\{0,1\}$. In the literature, such a BCS is also called a parity BCS. For convenience of a quantum generalization, it is equivalent to define the BCS over sign variables ranging in $\{+1,-1\}$. In this case, a Boolean function can be equivalently given by a multilinear function of a subset of variables on $\mathbb{R}$. For a linear BCS, the parity constraint becomes a product of the variables. In accordance with the notations in the main text, we mainly use the sign variables and denote each constraint $c_j$ as a multilinear function of a set of variables $\{v_i:i\in\mathcal{S}_j\}$, namely in the form of $\prod_{i\in\mathcal{S}_j}v_i=c_j\in\{+1,-1\}$. Nevertheless, it is sometimes more convenient to use the Boolean variables to represent a BCS. In correspondence to the sign variables,
\begin{equation}\label{eq:signtobool}
\begin{split}
v_iv_j=+1 \Leftrightarrow v_i\oplus v_j=0, \\
v_iv_j=-1 \Leftrightarrow v_i\oplus v_j=1, \\
\end{split}
\end{equation}
where the LHS are the notations using sign variables and the RHS are the notations using Boolean variables. We shall specify the notations if we resort to the Boolean variables.

If a BCS has a satisfying assignment, namely a fixed assignment to the variables that satisfies all the constraint, we say it has a classical solution. Note that if all the constraints are $c_j=+1$, the BCS can be trivially satisfied by assigning all the variables to be $+1$. With respect to the BCS size, searching for a classical solution to a general BCS is $\mathsf{NP}$-hard. On the other hand, the problem is in $\mathsf{P}$ for linear BCS, where one can apply Gaussian elimination or the replacement method to efficiently solve the system.

The quantum generalization of a BCS is an operator-valued constraint system. The variables $v_j$'s are replaced with linear operators $A_j$'s acting on a Hilbert space $\mathcal{H}$ with a finite dimension, such that
\begin{enumerate}
\item Each $A_j$ is Hermitian with eigenvalues in $\{1,-1\}$, i.e., $A_j=A_j^\dagger$ and $A_j^2=\mathbb{I}$ for all $j$.
\item $A_j$'s satisfy all the constraints with $c_i$ replaced with $c_i\mathbb{I}$, where $\mathbb{I}$ is the identity operator on $\mathcal{H}$.
\item If $A_i$ and $A_j$ appear in the same constraint, they commute with each other, i.e., $A_iA_j=A_jA_i$.
\end{enumerate}
If there exists a Hilbert space $\mathcal{H}$ with dimension $d$ and a set of linear operators following the above requirements, we say the BCS has a $d$-dimensional quantum satisfying assignment, or simply a quantum solution. As a side remark, the requirement that the operator variable acts on a finite-dimensional Hilbert space can be relaxed in several directions, including allowing an infinite dimension and limits of finite-dimensional systems. We do not discuss such generalizations and refer readers to Ref.~\cite{slofstra2019set,fu2021membership} for a more detailed definition.

If a BCS has a quantum solution, one can apply quantum measurements to realize it in an experiment, where they prepare independently and identically many copies of a quantum state and measure the observables in each constraint. For each constraint, the measurement results shall satisfy the relation in the fashion of classical variables. Note that the requirement for a quantum solution guarantees the validity of a joint measurement for each constraint. Due to the intrinsic randomness in quantum measurements, the same variable may take different values in different constraints. It is worth mentioning that measuring the set of observables on any state, including a maximally mixed state, generates the desired statistics for the constraints. The quantum satisfying assignment is also called a state-independent contextuality of the observables~\cite{cabello2008experimentally}. 

As the dimension can be arbitrary, searching for a quantum solution is undecidable~\cite{slofstra2019set,slofstra2020tsirelson}. A helpful way of thinking is to regard the BCS as a group presentation statement~\cite{cleve2014characterization,cleve2017perfect}. 

\begin{definition}[Group presentation]
Given a set $\mathcal{S}$, let $\mathcal{F}(\mathcal{S})$ be the free group on $\mathcal{S}$ and $\mathcal{R}$ a set of words on $\mathcal{S}$, and denote the quotient group of $\mathcal{F}(\mathcal{S})$ by the smallest normal subgroup containing each element in $\mathcal{R}$ as $\langle\mathcal{S}:\mathcal{R}\rangle$. A group $G$ is said to have the presentation $\langle\mathcal{S}:\mathcal{R}\rangle$ if it is isomorphic to $\langle\mathcal{S}:\mathcal{R}\rangle$.
\end{definition}

In the group presentation, the elements in $\mathcal{S}$ are called generators, and the elements in $\mathcal{R}$ are called relators.
Given a linear BCS, one can regard the constraints as a group presentation.

\begin{definition}[Solution group of a linear BCS]
Given a linear BCS with $n$ binary variables $\{v_i\}$ and $m$ constraints $\{c_j\}$, the solution group of the BCS is defined as the group with the following presentation:
\begin{equation}
\begin{split}
\Gamma=\{\{J,g_i:i\in[n]\}:\{g_i^2&=e,\forall i\in[n], \\
J^2&=e, \\
\forall j\in[m], g_kg_l&=g_lg_k, \forall k,l\in\mathcal{S}_j, \\
Jg_i&=g_iJ,\forall i\in[n], \\
\prod_{i\in\mathcal{S}_j}g_i&=J^{\chi(c_j=-1)},\forall j\in[m]\}\},
\end{split}
\end{equation}
where the group element $J$ corresponds to $-1$ in the BCS, $e$ defines the identity operator of the group, and $\chi(\cdot)$ is the indicator function that takes the value $1$ if the argument is true and $0$ if the argument is false.
\end{definition}

Note that if $J=e$, the solution group is trivial, as the assignment of all the group elements to be $e$ satisfies all the relators. On the contrary, should $J\neq e$, the solution group is non-trivial, and the group presentation corresponds to a valid operator-valued solution to the underlying BCS, as stated by the following lemma.

\begin{lemma}[\cite{cleve2017perfect,coladangelo2017robust}]
Given a linear BCS that defines a solution group with the group element $J$ corresponding to $-1$, if $J$ is non-trivial in some finite-dimensional representation of the solution group, then the BCS has a finite-dimensional quantum satisfying assignment. The converse is also true.
\label{lemma:nontrivialsolution}
\end{lemma}

The irreducible representation of the generators in the solution group determines the quantum realization~\cite{cleve2017perfect,paddock2023operator}. Here, a representation of the group refers to a group homomorphism of the group to a set of unitary operators on a Hilbert space, and an irreducible representation refers to a representation that does not have a non-trivial group-invariant subspace~\cite{serre1977linear}.
For a classical solution, it can be described by the Abelian group, where all the group elements commute with each other. 

As an example of linear BCS, we review the Mermin-Peres BCS, also widely known as the ``magic-square'' system~\cite{mermin1990simple,peres1990incompatible}. Note that one should not mistake the name ``magic'' for the quantum resource of magic states. The Mermin-Peres BCS involves $n=9$ variables and $m=6$ constraints. In terms of sign variables ranging in $\{+1,-1\}$, the BCS is defined as follows:
\begin{equation}
\begin{split}
v_1v_2v_3 &= 1, \\
v_4v_5v_6 &= 1, \\
v_7v_8v_9 &= 1, \\
v_1v_4v_7 &= 1, \\
v_2v_5v_8 &= 1, \\
v_3v_6v_9 &= -1.
\end{split}
\end{equation}
This BCS does not have a classical solution. On the other hand, it has a unique quantum solution over the Pauli group. We denote the operator that corresponds to $v_i$ as $A_i$ in accordance with the notations above. The quantum solution is given as follows~\cite{wu2016device,coladangelo2017robust}:
\begin{equation}\label{eq:MPSolution}
\begin{split}
A_1 &= \sigma_z\otimes\mathbb{I}, \\
A_2 &= \mathbb{I}\otimes\sigma_z, \\
A_3 &= \sigma_z\otimes\sigma_z, \\
A_4 &=\mathbb{I}\otimes\sigma_x, \\
A_5 &=\sigma_x\otimes\mathbb{I}, \\
A_6 &=\sigma_x\otimes\sigma_x, \\
A_7 &=\sigma_z\otimes\sigma_x, \\
A_8 &=\sigma_x\otimes\sigma_z, \\
A_9 &=\sigma_y\otimes\sigma_y,
\end{split}
\end{equation}
where $\otimes$ stands for the tensor product operation. 

Given a BCS, one can define an associated nonlocal game~\cite{cleve2014characterization}. In the nonlocal game, there are two cooperating players, Alice and Bob, who cannot communicate with each other once the game starts. With respect to a probability distribution, a referee randomly selects one constraint, $c_s$, and one variable, $v_t$, contained in the constraint. In our work, we always take the probability distribution to be uniform. The referee send $s$ to Alice and $t$ to Bob. Then, Alice returns an assignment to each variable $v_i$ in $c_s$ that satisfies the constraint and Bob returns an assignment to variable $v_t$. They win the game if and only if the assignments of Bob and Alice to $v_t$ are the same. This type of nonlocal game extends the well-known Clauser-Horne-Shimony-Holt (CHSH) game~\cite{clauser1969proposed}, where the underlying BCS involves two binary variables and two multi-linear constraints,
\begin{equation}
\begin{split}
v_1v_2 &= 1, \\
v_1v_2 &= -1.
\end{split}
\end{equation}
This BCS game is equivalent to the CHSH game in the sense of the probability distribution that Alice and Bob can achieve. Note that this BCS does not have either a classical solution or a quantum solution. Still, quantum strategies for the game can bring a higher winning probability than classical ones.

To maximize the winning probability, Alice and Bob can agree on a strategy for playing the game. We call a strategy is \emph{perfect} if it wins with probability 1. We say Alice and Bob apply a classical strategy if they can access only shared and local randomness. In quantum theory, Alice and Bob can pre-share entanglement and apply local quantum operations. In general, when a BCS does not have a solution, it is possible that Alice assigns different values to the same variable upon different questions of constraint. However, it brings limited advantages. For classical strategies, using basic linear algebra analysis, it is not hard to prove that a BCS game has a perfect classical strategy if and only if the corresponding BCS has a solution. It follows that to decide whether a general BCS game has a perfect classical strategy is in $\mathsf{NP}\text{-hard}$ for general Boolean constraints and in $\mathsf{P}$ for linear constraints with respect to the BCS size. In the above examples, the CHSH game has a maximal winning probability of $3/4$, and the Mermin-Peres game has a maximal winning probability of $8/9$. In the quantum case, if the BCS does not have an operator-valued solution, then there does not exist a perfect winning strategy for the associated nonlocal game, and \emph{vice versa}. This fact is first proved in Ref.~\cite{cleve2014characterization}. If the BCS has a quantum solution, it is linked to a perfect quantum winning strategy in a one-to-one correspondence. Suppose the quantum solution to the BCS is given by observables $\{A_i\}_i$ acting on a $d$-dimensional system. Alice and Bob first share a maximally entangled state, $\ket{\Phi^+}=\sum_{i=0}^{d-1}\ket{ii}/\sqrt{d}$. When the nonlocal game starts, upon receiving the constraint $c_s$, Alice measures the observables $A_i$ belonging to the constraint, and upon receiving the variable $v_t$, Bob measures the transpose of the observable $A_t$, denoted by $A_t^{\mathrm{T}}$. The measurement statistics satisfy the winning condition.

For the well-known existing BCS that are solvable, it either has a classical solution, which corresponds to an Abelian group, or a quantum solution of Pauli strings as in the case of the Mermin-Peres magic square, which corresponds to the Pauli group. In Ref.~\cite{arkhipov2012extending}, the author provides an efficient algorithm to determine perfect quantum solutions to a special type of linear BCS, where each variable shows up in exactly two constraints. Moreover, if such a BCS has a solution, the solution is necessarily given by Pauli strings, i.e., the solution is in the Pauli group. Following this result, it was conjectured that any linear BCS with an operator-valued satisfying assignment belongs to either of the two cases~\cite{arkhipov2012extending}. As mentioned in the main text, this conjecture has been suggested false~\cite{slofstra2019set,slofstra2020tsirelson}. In the next section, we directly disprove this conjecture with a specific linear BCS.

\section{Group-Theoretic Analysis of the Binary Constraint System}\label{supp:BCSResults}
In this section, we apply group-theoretic tools to analyse the properties of the proposed linear BCS in the main text. 
We first review the BCS proposed in this work. Given an undirected complete graph $G=(V,E)$ with $n$ vertices, it defines the following variables:
\begin{enumerate}
\item Each vertex $v\in V$ corresponds to one variable $a_v$.
\item Each undirected edge, denoted by $e=(u,v)\in E$, corresponds to three variables $x_{uv},y_{uv},z_{uv}$.
\item Every two disjoint edges, denoted by $e_1=(u,v)\in E$ and $e_2=(s,t)\in E$, where $s,t,u,v$ are different vertices, correspond to
\begin{enumerate}
\item one variable $b_{uv|st}\equiv b_{st|uv}$, where $b_{e_1e_2}=b_{e_2e_1}$;
\item two variables $c_{e_1e_2}\equiv c_{uv|st}$ and $c_{e_2e_1}\equiv c_{st|uv}$, where $c_{e_1e_2}\neq c_{e_2e_1}$ in general.
\end{enumerate}
\end{enumerate}
Based on these variables, the BCS contains the following constraints,
\begin{equation}\label{suppeq:PermutationBCS}
\begin{split}
a_u a_v y_{uv} &= 1, \forall (u,v)\in E, \\
x_{uv} y_{uv} z_{uv} &= 1, \forall (u,v)\in E, \\
x_{uv} x_{st} b_{uv|st} &= 1, \forall (u,v),(s,t)\in E, \\
x_{uv} z_{st} c_{uv|st} &= 1, \forall (u,v),(s,t)\in E, \\
b_{uv|st}b_{vs|ut}b_{su|vt} &=1, \forall (u,v),(s,t)\in E, \\
c_{uv|st}c_{vs|ut}c_{su|vt} &=1, \forall (u,v),(s,t)\in E, \\
\prod_{v\in V}a_v &= -1.
\end{split}
\end{equation}

For a nontrivial BCS, the system size is at least $n=4$. If the BCS has a solution, then the other variables can be generated by $a'$s and $x'$s. For convenience, we label the vertices with natural numbers from $1$ to $n$. For each vertex $v$ and each edge $(u,v),u\neq v$, we can use transpositions between elements in the set $[2n]$ to represent the generators, where $a_v=(2v-1\ 2v)$ and $x_{uv}=(2u-1\ 2v-1)(2u\ 2v)$. Here, $(i\ j)$ represents a transposition between $i$ and $j$. We have the following results for the BCS.

\begin{theorem}\label{thm:classicalSol}
For any BCS of the above form with $n\in 2\mathbb{N}+5=\{5, 7, 9, \cdots\}$, it has a classical solution.
\end{theorem}

\begin{proof}
When $n\in 2\mathbb{N}+5$, the BCS can be satisfied by assigning all the variables $a_v$'s to be $-1$ and all the other variables to be $1$.

\end{proof}

\begin{theorem}\label{thm:PauliSol}
When $n = 4$, this BCS has a two-qubit Pauli-string solution. On the other hand, the BCS does not have a classical solution or a single-qubit Pauli solution in this case.
\end{theorem}

\begin{proof}
The BCS having no classical solution can be directly checked by solving the BCS on the binary field. By using the fact that there is no state-independent contextuality in a qubit system, one can prove that the BCS does not have a single-qubit Pauli solution either~\cite{specker1960logik,kochen1967problem,renner2004quantum}. Later, we prove this statement under the context of linear BCS.    

For the former statement, here is one construction of the two-qubit Pauli-string solution. We abbreviate the Pauli operators $\mathbb{I},\sigma_x,\sigma_y,\sigma_z$ as $I,X,Y,Z$, respectively, and omit the tensor-product operator in the expressions.

\begin{table}[hbtp!]
\centering
\begin{tabular}{|c|c|c|c|c|c|}
\hline
\hline
$a_1$ & $a_2$ & $a_3$ & $a_4$ & &  \\ \hline
$-ZZ$ & $II$ & $ZI$ & $IZ$ & &  \\ \hline
\hline
$x_{12}$ & $x_{13}$ & $x_{14}$ & $x_{23}$ & $x_{24}$ & $x_{34}$ \\ \hline
$-YY$ & $XI$ & $IX$ & $II$ & $II$ & $ZZ$ \\ 
\hline
\hline
$y_{12}$ & $y_{13}$ & $y_{14}$ & $y_{23}$ & $y_{24}$ & $y_{34}$ \\ \hline
$-ZZ$ & $-IZ$ & $-ZI$ & $ZI$ & $IZ$ & $ZZ$ \\ 
\hline
\hline
$z_{12}$ & $z_{13}$ & $z_{14}$ & $z_{23}$ & $z_{24}$ & $z_{34}$ \\ \hline
$-XX$ & $-XZ$ & $-ZX$ & $ZI$ & $IZ$ & $II$ \\ 
\hline
\hline
$b_{12|34}$ & $b_{13|24}$ & $b_{14|23}$ & & &  \\ \hline
$XX$ & $XI$ & $IX$ & & &  \\ 
\hline
\hline
$c_{12|34}$ & $c_{13|24}$ & $c_{14|23}$ & $c_{34|12}$ & $c_{24|13}$ & $c_{23|14}$ \\ \hline
$-YY$ & $XZ$ & $ZX$ & $YY$ & $-XZ$ & $-ZX$ \\ 
\hline\hline
\end{tabular}
\caption{A two-qubit Pauli-string solution for the BCS when $n$ = $4$.}
\label{table:PauliSolution}
\end{table}

\end{proof}

\begin{theorem}[A special case of the results in Ref.~\cite{specker1960logik,kochen1967problem,renner2004quantum}]
If a linear BCS has a single-qubit operator-valued solution, then it has a classical solution.
\end{theorem}

\begin{proof}
Two-dimensional matrices have such a special property: Suppose $A,B\in\mathbb{C}^{2\times 2}$ and $[A,B]=AB-BA=0$. At least one of the following cases would happen (1) $A=c\mathbb{I}$ for some $c\in\mathbb{C}$; (2) $B=c\mathbb{I}$ for some $c\in\mathbb{C}$; (3) $A=cB$ for some $c\in\mathbb{C}$. 
Therefore, all two-dimensional matrices that are not proportional to the identity matrix $\mathbb{I}$ can be classified into different equivalence classes, with elements in a class proportional to each other. 

Given a single-qubit operator-valued solution, denote the first equivalent class as $\{c_1 O, c_2 O, \cdots | c_1, c_2, \cdots =\pm 1\}$ where $O\neq \pm\mathbb{I}$ and $O^2=\mathbb{I}$. Substituting $\{c_1 O, c_2 O, \cdots | c_1, c_2, \cdots =\pm 1\}$ with $\{c_1 \mathbb{I}, c_2 \mathbb{I}, \cdots | c_1, c_2, \cdots =\pm 1\}$ and not changing other variables also provides a solution. This is because (1) variables in different equivalent classes do not show up in the same constraint; (2) the constraints among the variables proportional to $\mathbb{I}$ and in the first equivalent class are not violated after the substitution. Similarly, we can set all variables proportional to $\mathbb{I}$ to give a solution. The proportional coefficient is a valid classical solution.

\end{proof}

\begin{theorem}\label{thm:noPauliSol}
For any BCS of the above form with $n\in 2\mathbb{N}+6=\{6, 8, 10, \cdots\}$, it does not have a Pauli-string solution.
\end{theorem}

\begin{proof}
Later we shall present a general method to determine whether a general linear BCS has a Pauli-string solution, where the current result can be regarded as a special case. Nevertheless, here we present a graph-based proof specific to this BCS, which is more illustrative.

We first specify the commutation properties of Pauli strings. The Pauli group elements are either anti-commuting, like $\{X,Z\}=XZ+ZX=0$, or commuting, like $[X\otimes X,Z\otimes Z]=(X\otimes X)(Z\otimes Z)-(Z\otimes Z)(X\otimes X)=0$. Therefore, if we swap any two operators in a multiplication of some Pauli-string operators, the operator value of the multiplication would at most differ with a sign.

Now we prove the theorem by contradiction. Assume the BCS described in the theorem has a Pauli-string solution. Without loss of generality, with respect to the correspondence between the BCS variables and vertices in a fully connected undirected graph, let us consider the sets of variables and constraints corresponding to a subgraph with five vertices, labeled with $1$ through $5$. For the quadrangle $\{1,2,4,5\}$, we have the equation $b_{12|45}b_{24|15}b_{41|25}=\mathbb{I}$. Substituting all the $b$-type variables by $x_{uv}x_{st} b_{uv|st}=\mathbb{I}$, we have
\begin{equation}
x_{12}x_{45}x_{24}x_{15}x_{41}x_{25}=\mathbb{I} \label{eq:bcs1245}.
\end{equation}
Similarly, for the quadrangles $\{1,3,4,5\}$ and $\{2,3,4,5\}$, we have the equations
\begin{align}
x_{13}x_{45}x_{34}x_{15}x_{41}x_{35}&=\mathbb{I}. \label{eq:bcs1345}\\
x_{23}x_{45}x_{34}x_{25}x_{42}x_{35}&=\mathbb{I}. \label{eq:bcs2345}
\end{align}
Next, by multiplying the left and right sides of Eqs.~\eqref{eq:bcs1245}\eqref{eq:bcs1345} and~\eqref{eq:bcs2345}, respectively, swapping the order of the variables, and eliminating the adjacent two variables that are the same, we get 
\begin{equation}
x_{12}x_{13}x_{23}x_{45}=\pm\mathbb{I},
\end{equation}
where ``$\pm$'' denotes that either case would happen. In this step, we have used the commutation properties of Pauli strings. In other words,
\begin{equation}
x_{45}\in\{\pm x_{12}x_{13}x_{23}\}.
\end{equation}
As a reminder, $x_{uv}$ and $x_{vu}$ represent the same variable. Note that there is nothing special about the choice of $x_{45}$ among the variables, and a similar result can be obtained with an arbitrary specification of a subgraph with five vertices. For instance, we can get
\begin{equation}
x_{46}\in\{\pm x_{12}x_{13}x_{23}\}.
\end{equation}
Combining the above two equations, we have $x_{45}=\pm x_{46}$. For $n\geq 6$, following the above procedure and enumerating all such identities, one shall find that all the $x$-type variables differ from each other up to a sign, i.e., $x_{uv}\in\{\pm x\}$ for all $(u,v)\in E$. Thus we can assume that $x_{uv}=x'_{uv}x$ where $x'_{uv}\in\{\pm 1\}$.

Following the specification of the $x$-type variables, for any four distinct vertices $u,v,s,t\in V$, we have the following expressions:
\begin{equation}
\begin{split}
b_{uv|st} &=x_{uv}x_{st}=x'_{uv}x'_{st}\mathbb{I}, \\
c_{uv|st} &=x_{uv}z_{st}=x_{uv}x_{st}y_{st}=x'_{uv}x'_{st}y_{st}=x'_{uv}x'_{st}a_{s}a_{t}.
\end{split}
\end{equation}
Consequently,
\begin{equation}
\begin{split}
b_{uv|st}b_{vs|ut}b_{su|vt}&=x'_{uv}x'_{st}x'_{vs}x'_{ut}x'_{su}x'_{vt} =1, \\
c_{uv|st}c_{vs|ut}c_{su|vt} &=x'_{uv}x'_{st}x'_{vs}x'_{ut}x'_{su}x'_{vt}a_{s}a_{t}a_{u}a_{t}a_{v}a_{t} =1. 
\end{split}
\end{equation}
Note that all the $a$-type variables commute with each other, since they simultaneously appear in the last equation of the BCS. We hence derive that $a_u a_v a_s a_t=1$ for all four distinct vertices $u,v,s,t\in V$, which is equivalent to
\begin{equation}
a_u=a_va_sa_t.
\end{equation}
By applying a similar argument as for the $x$-type variables, we shall find that all the $a$-type variables are identical, i.e., $a_v\equiv a,\forall v\in V$. This contradicts the constraint $\prod_{v\in V}a_v=-1$ when $n$, the number of vertices, is even. Therefore, the BCS of the above form with $n\in 2\mathbb{N}+6$ does not have a Pauli-string solution.

\end{proof}

\begin{theorem}\label{thm:groupSol}
For any BCS of the above form with $n\geq4$, label the vertices from $1$ to $n$. The BCS has a solution over the centralizer group of element $J=(1\ 2)(3\ 4)\cdots(2n-1\ 2n)$ in the permutation group $S_{2n}$.
\end{theorem}

\begin{proof}
Consider $x_{(u,v)} = (2u-1\ 2v-1)(2u\ 2v)$, $a_v = (2v-1\ 2v)$, and all the other variables generated by them. One can easily check that the assignment satisfies the constraints. As $J$ does not map to the identity element of the group, this is a non-trivial quantum solution following Lemma~\ref{lemma:nontrivialsolution}. The solution group corresponds to the centralizer of $J$ in the permutation group, given by
\begin{equation}
C_J = S_n\ltimes \mathbb{Z}_2^n,
\end{equation}
which is the semi-product of the permutation group $S_n$ generated by $x_{(u,v)}$, and an Abelian group $\mathbb{Z}_2^n$ generated by $a_v$. 

\end{proof}

Now we solve the irreducible representations of the solution, which gives the quantum realizations. We have the following theorem.

\begin{theorem}\label{thm:CJrep}
An irreducible representation of $S_n\ltimes \mathbb{Z}_2^n$ can be labeled by $\left(m, \theta^{(m)}, \rho^{(n-m)}\right)$ where $m$ is an integer in $[0, n]$, $\theta^{(m)}$ and $\rho^{(n-m)}$ are two irreducible representations of permutation groups $S_m$ and $S_{n-m}$, respectively. Given label $\left(m, \theta^{(m)}, \rho^{(n-m)}\right)$, we first get an irreducible representation of group $S_m\times S_{n-m}\ltimes \mathbb{Z}_2^n$, which is given by
\begin{equation}
\begin{split}
\phi^{(m, \theta, \rho)} : S_m\times S_{n-m}\ltimes \mathbb{Z}_2^n &\rightarrow M_{\dim \theta^{(m)}\times \dim \rho^{(n-m)}}(\mathbb{R})\\
x\in S_m, y\in  S_{n-m}, z\in \mathbb{Z}_2^n, (x,y,z) &\mapsto \theta^{(m)}(x)\otimes \rho^{(n-m)}(y) \otimes \varphi^{(m)}(z).
\end{split}
\end{equation}
Here, $\varphi^{(m)}$ is an irreducible representation of $\mathbb{Z}_2^n$ such that for any element $z = (z_1, z_2, \cdots, z_n) \in \mathbb{Z}_2^n$ where $\forall j, z_j\in \{0, 1\}$,
\begin{equation}
\varphi^{(m)}(z) = \prod_{j=1}^m (-1)^{z_j}.
\end{equation}
The irreducible representation of $S_n\ltimes \mathbb{Z}_2^n$ labeled by $\left(m, \theta^{(m)}, \rho^{(n-m)}\right)$, denoted as $\Phi^{(m, \theta, \rho)}$, is the induced representation of $\phi^{(m, \theta, \rho)}$. Specifically, one first finds the left coset of $S_m\times S_{n-m}\ltimes \mathbb{Z}_2^n$ in $S_n\ltimes \mathbb{Z}_2^n$, given by
\begin{equation}
\left\{g_1S_m\times S_{n-m}\ltimes \mathbb{Z}_2^n, g_2S_m\times S_{n-m}\ltimes \mathbb{Z}_2^n, \cdots, g_{\binom{n}{m}}S_m\times S_{n-m}\ltimes \mathbb{Z}_2^n\right\},
\end{equation}
where $\left\{g_1, g_2, \cdots, g_{\binom{n}{n}}\right\}$ are representative elements and $g_1$ is the identity. Then, the induced representation $\Phi^{(m, \theta, \rho)}$ is defined on the bases $\left\{g_i\Vec{\mathbf{e}}_j|1\leq i\leq \binom{n}{m}, 1\leq j\leq \dim \phi^{(m, \theta, \rho)}\right\}$ where $\Vec{\mathbf{e}}_j$ is a basis of the representation space of $\phi^{(m, \theta, \rho)}$. That is, for any element $g\in S_n\ltimes \mathbb{Z}_2^n$, suppose that $gg_i \in g_{\sigma(i)}S_m\times S_{n-m}\ltimes \mathbb{Z}_2^n$ and set $h_i = g_{\sigma(i)}^{-1}gg_i$ where $\sigma$ is a permutation on $\left\{1, 2,\cdots, \binom{n}{m}\right\}$, then
\begin{equation}
\Phi^{(m, \theta, \rho)}(g) = \left(\Pi_{\sigma}\otimes \mathbb{I}_{\dim \phi^{(m, \theta, \rho)}}\right)\left(\bigoplus_{i=1}^n \phi^{(m, \theta, \rho)}(h_i)\right).
\end{equation}
Here, $\Pi_{\sigma}$ is a permutation matrix defined on the computational basis $\ket{1}, \ket{2}, \cdots, \ket{\binom{n}{m}}$ and transforms $\ket{i}$ to $\ket{\sigma(i)}$.
\end{theorem}

\begin{proof}
This theorem is a direct corollary of Proposition 25 in~\cite{serre1977linear}. To get an irreducible representation of $C_J = S_n\ltimes \mathbb{Z}_2^n$, we start from the irreducible representation of $\mathbb{Z}_2^n$. Note that $\mathbb{Z}_2^n = \langle a_1\rangle \times \langle a_2\rangle \times \cdots \times \langle a_n\rangle$ is generated by $n$ two-order elements $a_1, a_2, \cdots, a_n$. Any irreducible representation of $\mathbb{Z}_2^n$ can be labeled by a vector with length $n$, like $(1, -1, \cdots, 1)$, denoting the values that $n$ generators would be mapped to in the representation. Meanwhile, we call two irreducible representations equivalent if they can be mutually transformed via $S_n$. In other words, two irreducible representations are equivalent if and only if the corresponding vectors have the same number of -1. To obtain an irreducible representation of $C_J$, we only need to consider inequivalent irreducible representations of $\mathbb{Z}_2^n$ under the transformation of $S_n$. Without loss of generality, we choose these irreducible representations as $(1, 1, \cdots, 1)$, $(-1, 1, \cdots, 1)$, $(-1, -1, \cdots, 1)$, and $\cdots$, $(-1, -1, \cdots, -1)$ and label them with the number of -1, that is, $\{0, 1, \cdots, n\}$.

For a number $m\in \{0, 1, \cdots, n\}$, we get an irreducible representation of $\mathbb{Z}_2^n$, denoted as $\varphi^{(m)}$, mapping the generators $a_i$ to $\varphi^{(m)}(a_i) = (-1)^{\mathbbm{1}_{i\leq m}}$. Then, we consider a subgroup of $S_n$, such that any element $h$ in this subgroup satisfies $\forall z \in \mathbb{Z}_2^n$,
\begin{equation}\label{eq:invariant}
\varphi^{(m)}(gzg^{-1}) = \varphi^{(m)}(z),
\end{equation}
or equivalently, $1\leq i\leq n$,
\begin{equation}
\varphi^{(m)}(ga_ig^{-1}) = \varphi^{(m)}(a_i).
\end{equation}
Clearly, this subgroup must be $S_m\times S_{n-m}$. Then, one can define the irreducible representation of $S_m\times S_{n-m}\ltimes \mathbb{Z}_2^n$ as
\begin{equation}
\begin{split}
\phi^{(m, \theta, \rho)} : S_m\times S_{n-m}\ltimes \mathbb{Z}_2^n &\rightarrow M_{\dim \theta^{(m)}\times \dim \rho^{(n-m)}}(\mathbb{R})\\
x\in S_m, y\in  S_{n-m}, z\in \mathbb{Z}_2^n, (x,y,z) &\mapsto \theta^{(m)}(x)\otimes \rho^{(n-m)}(y) \otimes \varphi^{(m)}(z),
\end{split}
\end{equation}
where $\theta^{(m)}$ and $\rho^{(n-m)}$ are two irreducible representations of permutation groups $S_m$ and $S_{n-m}$, respectively. Note that $\phi^{(m, \theta, \rho)}$ is a well-defined group homomorphism due to the condition of Eq.~\eqref{eq:invariant}. Proposition 25 in~\cite{serre1977linear} tells us that any irreducible representation of $S_n\ltimes \mathbb{Z}_2^n$ can be constructed by the induced representation of $\phi^{(m, \theta, \rho)}$ by traversing $m$, $\theta$, and $\rho$. Proof is done.

\end{proof}

For a perfect strategy of the non-local game, the element $J$ must be mapped to a non-identity element. Note that any element in $C_J$ commutes with $J$. Via Theorem~\ref{thm:CJrep}, one can obtain the following result:
\begin{equation}
\begin{split}
\Phi^{(m, \theta, \rho)}(J) &= \varphi^{(m)}(J)\mathbb{I}_{\dim \Phi}\\
&= \varphi^{(m)}\left(\prod_{i=1}^n a_i\right)\mathbb{I}_{\dim \Phi}\\
&= \prod_{i=1}^n\varphi^{(m)}(a_i)\mathbb{I}_{\dim \Phi}\\
&= (-1)^{\mathrm{mod}(m, 2)}\mathbb{I}_{\dim \Phi},
\end{split}
\end{equation}
where $\dim \Phi = \binom{n}{m}\dim \phi = \binom{n}{m}\dim \theta^{(m)} \dim \rho^{(n-m)}\geq \binom{n}{m}$. Thus, $\Phi^{(m, \theta, \rho)}(J)$ corresponds to a perfect measurement strategy if and only if $m$ is odd. The smallest dimension of the quantum system for a perfect strategy is $n$ when $m=1$ or $m=n-1$ and $\dim \theta^{(m)} = \dim \rho^{(n-m)} = 1$.

The quantum realization of the BCS in Eq.~\eqref{suppeq:PermutationBCS} is not unique. The underlying reason is that unlike the Pauli group, the permutation group has more than one inequivalent irreducible representations~\cite{serre1977linear}. For $n=8$, which is nearly the smallest size for a non-trivial result where there is not a Pauli-string solution to the BCS, we consider the case where $m = 1$ and $\theta$ and $\rho$ are both trivial representations, in which the dimension of the quantum system is $8$. It implies that the corresponding non-local game can be realized with only $3$ EPR pairs. The representations of the generators are given by the following:
\begin{align}
&a_1 = \begin{pmatrix}
-1 & 0 & \cdots & 0\\
0 & 1 & \cdots & 0\\
\vdots & \vdots & \ddots & \vdots\\
0 & 0 & \cdots & 1
\end{pmatrix},
a_2 = \begin{pmatrix}
1 & 0 & \cdots & 0\\
0 & -1 & \cdots & 0\\
\vdots & \vdots & \ddots & \vdots\\
0 & 0 & \cdots & 1
\end{pmatrix},
\cdots,
a_8 = \begin{pmatrix}
1 & 0 & \cdots & 0\\
0 & 1 & \cdots & 0\\
\vdots & \vdots & \ddots & \vdots\\
0 & 0 & \cdots & -1
\end{pmatrix},\\
&x_{12} = \begin{pmatrix}
0 & 1 & 0 & \cdots & 0\\
1 & 0 & 0 & \cdots & 0\\
0 & 0 & 1 & \cdots & 0\\
\vdots & \vdots & \vdots & \ddots & \vdots\\
0 & 0 & 0 & \cdots & 1
\end{pmatrix}, 
x_{13} = \begin{pmatrix}
0 & 0 & 1 & \cdots & 0\\
0 & 1 & 0 & \cdots & 0\\
1 & 0 & 0 & \cdots & 0\\
\vdots & \vdots & \vdots & \ddots & \vdots\\
0 & 0 & 0 & \cdots & 1
\end{pmatrix}, 
\cdots,
x_{18} = \begin{pmatrix}
0 & 0 & 0 & \cdots & 1\\
0 & 1 & 0 & \cdots & 0\\
0 & 0 & 1 & \cdots & 0\\
\vdots & \vdots & \vdots & \ddots & \vdots\\
1 & 0 & 0 & \cdots & 0
\end{pmatrix}.
\end{align}
The value of $x_{uv}$ equals $x_{1u}x_{1v}x_{1u} = \mathbb{I}_8-\mathbf{e}_{uu}-\mathbf{e}_{vv}+\mathbf{e}_{uv}+\mathbf{e}_{vu}$, where $\mathbb{I}_8$ is an eight-dimensional identity operator, and $\mathbf{e}_{ij}$ denotes an elementary matrix, of which the element in the $i$'th row and $j$'th column is one, and all the other elements are zero.

From the expression of the generators, we can see that the measurement observables do not belong to the Pauli group and need magic to realize.
Following the same derivation, one can prove that the smallest non-trivial irreducible representation of the solution to the BCS defined over $n$ vertices requires an $n$-dimensional system. Thus, we obtain an upper bound of the smallest number of qubits to win the non-local game.

\begin{corollary}
The smallest number of qubits to win the associated nonlocal game of Eq.~\eqref{suppeq:PermutationBCS} is $O(\log n)$.
\end{corollary}

\section{Capabilities of Clifford strategies in the Nonlocal Game}\label{appendsc:clifford}
In proving the ``magic'' advantage in shallow circuit quantum computation, we need to specify the capabilities of Clifford strategies in the nonlocal BCS game. Thanks to the algebraic structure of BCS, we can use mature techniques from linear algebra to obtain quantitative results.

Suppose the players in a nonlocal game are restricted to Clifford operations only, or that they do not have access to quantum magic resources. In this case, the most general strategy they can apply to playing the nonlocal game is as follows:
\begin{itemize}
\item Before the nonlocal game starts:
\begin{enumerate}
\item Alice and Bob prepares an $n$-qubit state and initialize it in $\ket{0}$.
\item Alice and Bob apply joint Clifford operations and Pauli-string measurements to the state and evolve it into an entangled state $\rho_{AB}$, where the subscripts denote the subsystems they each will hold in the game.
\end{enumerate}
\item After the nonlocal game starts: Alice and Bob each applies Pauli-string measurements to their own quantum system.
\end{itemize}

In our discussions, we allow an arbitrarily large $n$. Using a convexity argument, we know that a mixed state does not bring any advantage to Alice and Bob in winning the nonlocal game, and we can hence take $\rho_{AB}$ as a pure state $\ket{\psi}$ without loss of generality. By further applying the Schmidt decomposition result, a pure state can be written as
\begin{equation}
\ket{\psi}=\sum_{i=0}^{d-1}\alpha_i\ket{ii},
\end{equation}
where $\forall i, \alpha_i\geq0$, and $\sum_{i=0}^{d-1}\alpha_i^2=1$. Note that the maximally entangled state, which is 
\begin{equation}\label{eq:MaxEnt}
\ket{\Phi^+}=\sum_{i=0}^{d-1}\frac{1}{\sqrt{d}}\ket{ii},
\end{equation}
can be prepared by applying control-NOT operations to $\ket{0}$, which is a Clifford operation. Therefore, a general bipartite entangled state $\ket{\psi}$ shared by Alice and Bob can only be linked with $\ket{\Phi^+}$ with a Clifford operation, i.e., $\ket{\psi}=U_C\ket{\Phi^+}$.

Based upon the above observations, we discuss the capabilities of Clifford operations in playing a parity BCS nonlocal game. In demonstrating the magic advantage, we are interested in the parity BCS that do not have a Pauli-string quantum satisfying assignment. We have the following result for these instances.

\begin{theorem}\label{thm:winprob}
Suppose a parity BCS does not have a satisfying assignment with Pauli-string observables. Then, for any Clifford strategy, there exist a constraint labelled by $c_s$ and a variable in it labelled by $v_t$, where the probability that the assignments of Alice and Bob to $v_t$ under the constraint $c_s$ are identical does not exceed $1/2$.
\end{theorem}

\begin{proof}
In the first part of the proof, we prove the case where Alice and Bob initially share a maximally entangled state $\ket{\Phi^+}$ in Eq.~\eqref{eq:MaxEnt} of an arbitrary dimension and then generalize the result to general Clifford strategies. In the BCS nonlocal game, without loss of generality, upon receiving the constraint labelled by $s$, Alice shall measure a set of commuting Pauli-string observables that return a satisfying assignment to the constraint, since a failure in satisfying the constraint results in a loss in the nonlocal game. On the other hand, the observables that she measures for the same variable, e.g., $v_t$, in different constraints can be different. To specify her strategy, we denote the observable Alice measures for vairable $v_t$ in the constraint $c_s$ as $A_t^{(s)}$. On Bob's side, we denote the observable he measures for variable $v_t$ as $B_t$. 

Now, suppose Alice and Bob initially share the maximally entangled state $\ket{\Phi^+}$. Assume there exists a Clifford strategy, such that $\forall s,t$,
\begin{equation}
\bra{\Phi^+}A_t^{(s)}\otimes B_t\ket{\Phi^+}>0.
\end{equation}
Since Alice and Bob apply a Clifford strategy, $A_t^{(s)}$ and $ B_t$ are both Pauli strings, so is $A_t^{(s)} B_t^\mathrm{T}$. 
Using the property of $\ket{\Phi^+}$, the left-hand side of the above equation equals $\bra{\Phi^+}A_t^{(s)} B_t^\mathrm{T}\otimes \mathbb{I}\ket{\Phi^+}=\tr(A_t^{(s)} B_t^\mathrm{T})/d$, where $d$ is the system dimension. Since $A_t^{(s)} B_t^\mathrm{T}$ is a Pauli string, we have $\tr(A_t^{(s)} B_t^\mathrm{T})/d\in\{0,\pm 1\}$. According to our assumption, we conclude that $A_t^{(s)}=B_t^\mathrm{T}$ and $\bra{\Phi^+}A_t^{(s)}\otimes B_t\ket{\Phi^+}=1$ for all $s,t$. Besides,
\begin{equation}\label{eq:ApproxIdentity}
\begin{split}
\bra{\Phi^+}A_t^{(s_1)}A_t^{(s_2)}\otimes\mathbb{I}\ket{\Phi^+} &=\bra{\Phi^+}(A_t^{(s_1)}\otimes\mathbb{I})(A_t^{(s_2)}\otimes\mathbb{I})\ket{\Phi^+} \\
&=\bra{\Phi^+}(A_t^{(s_1)}\otimes\mathbb{I})(A_t^{(s_2)}\otimes \mathbb{I})\ket{\Phi^+}\bra{\Phi^+}(\mathbb{I}\otimes B_t)(\mathbb{I}\otimes B_t)\ket{\Phi^+} \\
&\geq\bra{\Phi^+}A_t^{(s_1)}\otimes B_t\ket{\Phi^+}\bra{\Phi^+}A_t^{(s_2)}\otimes B_t\ket{\Phi^+} \\
&>0,
\end{split}
\end{equation}
which holds for all $s_1,s_2$ and $t$. In the third line, we apply the Cauchy-Schwarz inequality. The only value that the above equation can take is hence $1$, indicating that $A_t^{(s_1)}=A_t^{(s_2)}$. Thus we can omit the superscript $(s)$.

Since the linear BCS does not have a satisfying assignment with Pauli-string observables, we can use the commutation properties of Pauli operators and the substitution method and derive an expression $A_{t_1}A_{t_2}\cdots A_{t_m}=-\mathbb{I}$ for a set of variables that leads to a contradiction of $\mathbb{I}=-\mathbb{I}$. The proof of this statement shall be given in Corollary~\ref{corollary:substitution} in Appendix~\ref{algo}. Therefore, for any Clifford strategy, there exists a particular pair of inputs $s,t$ such that $\bra{\Phi^+}A_t^{(s)}\otimes B_t\ket{\Phi^+}\leq 0$. Consequently,
\begin{equation}
\Pr[\text{win}\mid s,t]=\frac{1}{2}\left(1+\bra{\Phi^+}A_t^{(s)}\otimes B_t\ket{\Phi^+}\right)\leq\frac{1}{2}.
\end{equation}
Therefore, the average winning probability of the game is
\begin{equation}
\Pr[\text{win}]=\sum_{s',t'}\Pr[\text{win}|s',t']\Pr[s',t']\leq 1-\frac{1}{2}\Pr[s,t].
\end{equation}

In the second part of the proof, we use the definition of Clifford operations that map a Pauli-string observable to a Pauli-string observable. For any initial state $\ket{\psi}$ that Alice and Bob may share in advance, it is linked with $\ket{\Phi^+}$ via a Clifford operation $U_C$. Then for any Pauli-string observables $A_t,B_t$,
\begin{equation}
\bra{\psi}A_t\otimes B_t\ket{\psi}=\bra{\Phi^+}U_C^\dag(A_t\otimes B_t)U_C\ket{\Phi^+}\equiv\bra{\Phi^+}A_t'\otimes B_t'\ket{\Phi^+},
\end{equation}
where $A_t'$ and $B_t'$ are also Pauli-string observables that adapt to the systems of Alice and Bob, respectively. Then, either $\{A_t'\}$ fails in yielding a satisfying assignment to one of the constraints, or the proof dates back to the first part. This finishes the proof.

\end{proof}

\section{1D-Magic-BCS-Based Computational Problems}
\label{supp:ShallowCircuit}

In this section, we introduce the computational tasks in detail by embedding the BCS nonlocal game into a one-dimensional grid. We consider two kinds of tasks. In the first task, we transform the BCS nonlocal game into a two-round interactive relation problem. In the second task, we introduce a sampling problem, while the distribution of the circuit output has to satisfy particular conditions. We will show these tasks can be solved by a fixed generic constant-depth quantum circuit with only bounded fan-in gates. On the contrary, a magic-free probabilistic circuit for the first task, or a fixed one for the second task, requires a circuit depth that increases at least logarithmically to the input size. For simplicity, we consider the non-trivial BCS game with size $n=8$. Three pairs of qubits suffice to realize the nonlocal game associated with this BCS, as shown in Appendix~\ref{supp:BCSResults}. One can consider other values of $n$, where the proofs are similar.

In the shallow circuit computation, we apply the modified BCS given in Methods. That is, the constraint of $\prod_{v\in\mathcal{V}}a_v=-1$ is replaced with the set of constraints
\begin{equation}
\left\{\begin{array}{rl}
a_1a_2a_{12}&=1 \\
a_{12}a_3a_{123}&=1 \\
\cdots \\
a_{1\cdots n-3}a_{n-2}a_{1\cdots n-2}&=1 \\
a_{1\cdots n-2}a_{n-1}a_n&=-1.
\end{array}\right.
\end{equation}
For the modified BCS with size $n=8$, by applying Theorem~\ref{thm:winprob}, we have the following result.

\begin{lemma}
In the modified BCS game with $n=8$, suppose the questions are picked up uniformly at random. Then, the maximal winning probability for all Clifford strategies is upper-bounded by
\begin{equation}
p_{\mathrm{Clif}}\leq 1-\frac{1}{6252}.
\end{equation}
\label{lemma:modifiedwinprob}
\end{lemma}

Below, we introduce two kinds of computational tasks and prove the hardness of shallow magic-free circuits in solving these tasks.

\subsection{Two-round interactive relation problem}\label{supp:interactive}

Now we introduce the two-round interactive relation problem $R_N$, which is labeled with a number $N$ representing the problem size. The problem is divided into two rounds, and each round is a relation problem. One can assume that two players, Alice and Bob, collaborate with each other to solve $R_N$. The input and output formats of each round of $R_N$ are given below.
\begin{enumerate}
\item Input: in each round, Alice and Bob are given a question,
\begin{equation}\label{eq:RninputInteractive}
q=(\alpha_1,\beta_1,\cdots,\alpha_N,\beta_N),
\end{equation}
where $q$ stands for ``question'', $\alpha_i$ is the input on Alice's side at site $i\in[N]$, and $\beta_i$ is the input on Bob's side at site $i\in[N]$. 
It is required that $\alpha_i,\beta_i\in\{\bot,\top\}$ for the first round and $\alpha_i\in\{\pm 1\}^6\times(\mathcal{Q}^A\cup\{\bot\}),\beta_i\in\mathcal{Q}^B\cup\{\bot\}$ for the second round, where $\mathcal{Q}^A$ consists of the set of constraints in the BCS, and $\mathcal{Q}^B$ consists of the set of variables in the BCS. Here, $\bot$ represents a null input and $\top$ is another character different from $\bot$.
\item Output: in each round, Alice and Bob need to return a reaction to the question,
\begin{equation}\label{eq:RnoutputInteractive}
r=(\mathbf{r}_1^A,\mathbf{r}_1^B,\cdots,\mathbf{r}_N^A,\mathbf{r}_N^B),
\end{equation}
where $r$ stands for ``reaction'', $\mathbf{r}_i^A=(r_i^A(1),r_i^A(2),r_i^A(3))\in\{\pm1\}^3$ is the output on Alice's side at site $i$, and similarly on Bob's side.
\end{enumerate}

Now, we define the 1D magic BCS relation problem $R_N$. In the computational task, Alice and Bob are promised to receive an instance given by a tuple $(j,k,\alpha,\beta),1\leq j<k\leq N$.

For the first round of $R_N$, the input is 
\begin{equation}\label{eq:shallowgameRound1}
\alpha_i=\begin{cases}
\top, & \text{if $i=j$}, \\
\bot, & \text{if $i\neq j$},
\end{cases}
\quad
\beta_i=\begin{cases}
\top, & \text{if $i=k$}, \\
\bot, & \text{if $i\neq k$}.
\end{cases}
\end{equation}
We do not require the output of the first round to satisfy a specific relation. Nevertheless, it will be related to the input in the second round. Suppose the output of the first round is $r^{(1)}=(\mathbf{r}^{(1)A}_1,\mathbf{r}^{(1)B}_1,\cdots,\mathbf{r}^{(1)A}_N,\mathbf{r}^{(1)B}_N)$. Then, based on $r^{(1)}$, the input for the second round is given by
\begin{equation}\label{eq:shallowgameRound2}
\alpha_i=\begin{cases}
\left(\mathbf{p}^A, \mathbf{p}^B, \alpha\right), & \text{if $i=j$}, \\
\left(\mathbf{1}, \bot\right), & \text{if $i\neq j$},
\end{cases}
\quad
\beta_i=\begin{cases}
\beta, & \text{if $i=k$}, \\
\bot, & \text{if $i\neq k$},
\end{cases}
\end{equation}
where
\begin{equation}
    \mathbf{p}^A(l) = \prod_{i=j+1}^{k}r^{(1)A}_i(l),\quad
    \mathbf{p}^B(l) = \prod_{i=j}^{k-1}r^{(1)B}_i(l),\quad
    l\in\{1,2,3\},
\end{equation}
and $\mathbf{1}$ represents the vector with all elements $1$.
That is, we require the sites $j$ on Alice’s side and $k$ on Bob’s side to play the BCS nonlocal game with questions $\alpha$ and $\beta$, respectively, based on the auxiliary information $(\mathbf{p}^A, \mathbf{p}^B)$ sending to Alice. Alice and Bob are required to give an output satisfying the following requirement:
\begin{equation}\label{eq:successrelationRound2}
(\mathbf{r}_j^A,\mathbf{r}_k^B)=f(\alpha,\beta),
\end{equation}
where $f$ is the relation defined by the BCS nonlocal game.

Below, we show that the both rounds of the 1D magic relation problem $R_N$ can be solved by a \QNCzero circuit, but the second round cannot be solved by any \ClifNCzero circuit. We first show that there exists a shallow circuit with generic bounded fan-in quantum gates that perfectly completes this task. Now consider the following strategy (Steps 1,2 for the first round and the rest for the second):
\begin{enumerate}
\item Alice and Bob share $3N$ pairs of EPR states, $\ket{\Phi^+}^{\otimes 3N}$, where $\ket{\Phi^+}=(\ket{00}+\ket{11})/\sqrt{2}$, and arrange them in three layers, denoted by $q_{2i-1}(l),q_{2i}(l),i\in[N],l\in\{1,2,3\}$, where qubits $q_{2i-1}(l)$ and $q_{2i}(l)$ reside in the state $\ket{\Phi^+}$. Alice holds the qubits $q_{2i-1}(l)$ and Bob holds the qubits $q_{2i}(l)$.
\item For any $j\leq i\leq k-1$, perform an entanglement swapping operation between pairs of EPRs with a BSM on qubits $q_{2i}(l)$ and $q_{2i+1}(l)$. Denote the Bell state measurement results on the pair of adjacent qubits $q_{2i}(l)$ and $q_{2i+1}(l)$ as $r^{(1)B}_i(l)$ and $r^{(1)A}_{i+1}(l)$, respectively, and as the output of the first round.
\item Based on the input of $\prod_{j<t\leq k}\mathbf{r}^{(1)A}_t, \prod_{j\leq t<k}\mathbf{r}^{(1)B}_t$ from the second round, which is the information from the BSM measurement result, Alice performs a local operation on $q_{2j-1}(l)$ to correct the quantum state on $q_{2j-1}(l)$ and $q_{2k}(l)$ to the EPR state $\ket{\Phi^+}$. Specifically, $\prod_{j<t\leq k}r^{(1)A}_t(l)$ and $\prod_{j\leq t<k}r^{(1)B}_t(l)$, which is part of the first $6$ bits of $\alpha_j$, determine which states $q_{2j-1}(l)$ and $q_{2k}(l)$ are among the four Bell states $\{\ket{\Phi^+},\ket{\Phi^-},\ket{\Psi^+},\ket{\Psi^-}\}$, and Alice can perform the corresponding operation $\{I,Z,X,XZ\}$ to correct it to $\ket{\Phi^+}$.
\item On the three pairs of qubits $q_{2j-1}(l)$ and $q_{2k}(l)$, Alice and Bob perform the measurements corresponding to the winning strategy in the BCS nonlocal game with respect to questions $(\alpha,\beta)$ and obtain outputs $(r_j^A(l),r_k^B(l))$.
\item Take an arbitrary measurement on the qubits that are not measured and record the measurement results with respect to the site indices.
\end{enumerate}

One can see the reason that we modify the underlying BCS game: Alice needs to output an assignment to all the variables that appear in the constraint. As Alice can only output $3$ bits in the relation problem, we need to decompose the original $n$-variable constraint into smaller ones. Measuring the observables in the winning strategy of the BCS nonlocal game thus results in the desired statistics. Moreover, the above strategy can be realized in a constant depth with finite fan-in operations and a computational basis measurement. Nevertheless, there is a minimal fan-in size of the gates to realize the above strategy. The following theorem gives a sufficient gate fan-in size.

\begin{theorem}\label{thm:K}
Suppose the quantum computation circuits act on qubits or bits. Then, the fan-in size $K=14$ is sufficient for the above strategy.
\end{theorem}

\begin{proof}
As counted in Methods, in the modified BCS nonlocal game for the relation problem $R^8_N$, the number of variables is $727+1<2^{10}$, and the number of constraints is $1042+1<2^{11}$, where we also need to consider the null input $\bot$. Thus, all possible inputs $\alpha,\beta$ can be encoded as $11$ bits.

Next, we prove that all the operations in the above strategy can be implemented by $K$-bounded fan-in gates with $K=14$. Note that all constant input size boolean function can be computed in \NCzero, so we only care about the classically controlled quantum gates [Fig.~\ref{fig:bounded}(b)] and quantum measurements [Fig.~\ref{fig:bounded}(c)]. We analyze the algorithm step by step:
\begin{enumerate}
\item EPR preparation: No classical bit is involved. The quantum circuit involves simply one or two-qubit quantum gates, hence $K\geq 0+2=2$ in this step.
\item BSM: First determine whether the input is $\bot$ through classical computing. Then perform BSM if it so and do nothing if not. This step requires $K\geq 1+2=3$.
\item EPR correction: 2 classical bits to control single qubit quantum gates. This step requires $K\geq 2+1=3$.
\item Nonlocal game: 11 classical bits are needed to control 3-qubit quantum gates. This step requires $K\geq 11+3=14$.
\end{enumerate}
Therefore, all the operations in the shallow-circuit strategy to solve the computational problem can be implemented using $K$ bounded fan-in gates with $K=14$. This finishes the proof.

\end{proof}

As a side remark, note that a classically controlled quantum gate with a constant number of classical control bits can be decomposed into compositions of single-bit classically controlled quantum gates within a constant depth where the quantum part acts on up to two qubits~\cite{nielsen2010quantum}. Thus, actually, $K=3$ is sufficient, albeit a compromise of a few circuit layers.

\begin{figure*}[hbt!]
\centering
\begin{subfigure}[b]{0.32\textwidth}
    \centering
    \includegraphics[scale=0.31]{Figure/bounded_fan_in_round.pdf}
    \caption{}
\end{subfigure}
\hfill
\begin{subfigure}[b]{0.32\textwidth}
    \centering
    \includegraphics[scale=0.31]{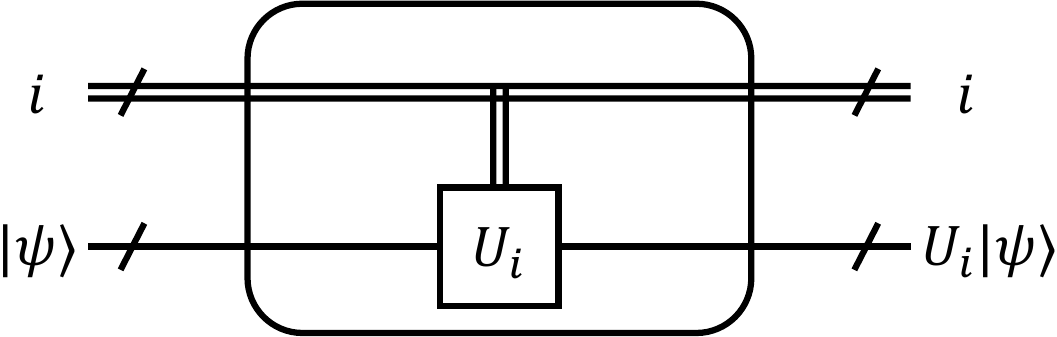}
    \caption{}
\end{subfigure}
\hfill
\begin{subfigure}[b]{0.32\textwidth}
    \centering
    \includegraphics[scale=0.31]{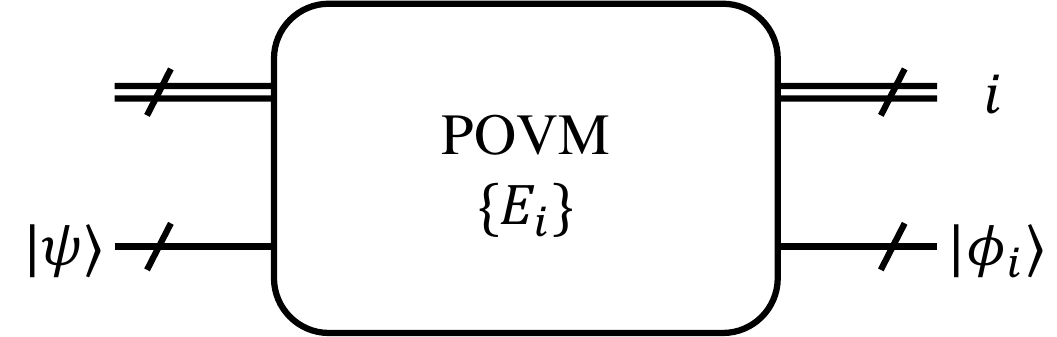}
    \caption{}
\end{subfigure}

\captionsetup{justification=Justified,singlelinecheck=false}
\caption{(a) A general $K$-bounded fan-in gate acting on $n_c$ bits and $n_q$ qubits, with $n_c+n_q\leq K$. There are two types of gates with $n_q > 0$: (b) Classically controlled quantum gates. The classical input $i$ controls whether to apply the quantum gate $U_i$ to the quantum input state $\ket{\psi}$; (c) Quantum measurement characterized by a positive operator-valued measure (POVM). The classical system acts as a register to record the measurement result. In a \ClifNCzero circuit, $U_i$ is restricted to a Clifford operation. The controlled operation in this circuit requires the classical part to perform the logical operations of parity and negation and the quantum part to perform Pauli operations so that the intermediate measurements are movable to the end of the circuit~\cite{delfosse2021bounds}.
Meanwhile, in a \ClifNCzero circuit, each measurement element $E_i$ should be a mixture of stabilizer states. Equivalently, for a magic-free quantum input state $\ket{\psi}$, the post-measurement state $\ket{\phi_i}$ in a \ClifNCzero circuit still does not contain quantum magic. In this work, we simply consider projective measurements. Both (b) and (c) are special cases of (a).
} 
\label{fig:bounded}
\end{figure*}

Below, we prove the hardness of the problem for a \ClifNCzero circuit with only bounded fan-in classical gates and Clifford gates, including classically controlled Clifford gates and constant-weight Pauli-string measurements. We illustrate these two types of bounded fan-in gates in Fig.~\ref{fig:bounded}. The bounded fan-in classically controlled quantum gates require the numbers of classical input bits and the controlled qubits to be both finite, and the controlled gate is Clifford in a \ClifNCzero circuit. The constant-weight Pauli measurement measures a Pauli observable on a constant number of qubits, where the POVM element is a stabilizer state. As a constant-weight Pauli measurement is equivalent to implementing a constant-depth Clifford gate followed by the computational basis measurement on the first qubit and implementing the inverse of the Clifford gate, we can take all measurements as computational basis measurements, or $\sigma_z$ measurements. Meanwhile, one can always consider the Pauli measurements at the end of the circuits by moving intermediate measurements backward~\cite{delfosse2021bounds}.

Now, consider a circuit with depth $D$, and the gates within it have fan-in bounded by $K$, which means the total number of input classical bits and qubits of the gates is no larger than $K$. Denote the qubit or bit value at index $v$ as $i_v$ and suppose $\mathcal{E}_1$ is the gate of the first layer that contains $i_v$ as an input. Then, $\mathrm{supp}(\mathcal{E}_1(i_v))$ determines a set of qubits and bits after the first layer of the circuit that may be affected by $i_v$. Similarly, we can consider the qubits and bits that may be affected in the next layer of the circuit. Denote the gates in each circuit layer are given by $\{\mathcal{E}_1, \mathcal{E}_2, \cdots, \mathcal{E}_D\}$. In the end, we call the set of qubits and bits
\begin{equation}
L_C^{\rightarrow}(i_v)=\mathrm{supp}(\mathcal{E}_D(\mathrm{supp}(\cdots\mathrm{supp}(\mathcal{E}_2(\mathrm{supp}(\mathcal{E}_1(i_v)))))))
\end{equation}
the forward light cone of $i_v$.

The backward light cone of an output bit or qubit, $o_w$, at index $b$ can be defined with the reverse of the forward light cone of the input and given by
\begin{equation}
L_C^{\leftarrow}(o_w) := \{i_v| o_w\in L_C^{\rightarrow}(i_v)\}.
\end{equation}
The backward light cone of an output set, $O$, is defined as
\begin{equation}
L_C^{\leftarrow}(O) := \bigcup_{o\in O} L_C^{\leftarrow}(o).
\end{equation}
Note that if a depth-$D$ quantum circuit, $\mathcal{C}$, only comprises gates with fan-in bounded by $K$, then
\begin{equation}
|L_C^{\leftarrow}(o)| \leq K^D,
\end{equation}
and
\begin{equation}
|L_C^{\leftarrow}(O)| \leq |O|K^D.
\end{equation}

Note that the input of the problem $R_N$ is classical and given by $q = (\alpha_1,\beta_1,\cdots,\alpha_N,\beta_N)$. Before acting gates, there are also an arbitrary number of classical ancillas with value $x$ and quantum ancillas at state $\ket{\psi}$, which do not contain any input information. With the circuit's evolution, the input information will spread among classical bits and qubits via classically controlled quantum gates. Nonetheless, due to the gates being bounded fan-in and the circuit being at constant depth, the input information cannot spread a lot. The output of $R_N$ can be read out from the classical bits after constant-depth circuit evolution. Without loss of generality, we can assume the output of the second round of $R_N$ is read out from the first $6N$ bits at the last step of the circuit, which is $r=(\mathbf{r}_1^A,\mathbf{r}_1^B,\cdots,\mathbf{r}_N^A,\mathbf{r}_N^B)$ as mentioned above. We depict this procedure in Fig.~\ref{fig:lightcone}(a).

\begin{figure*}[hbt!]
\centering 
\includegraphics[width=.9\textwidth]{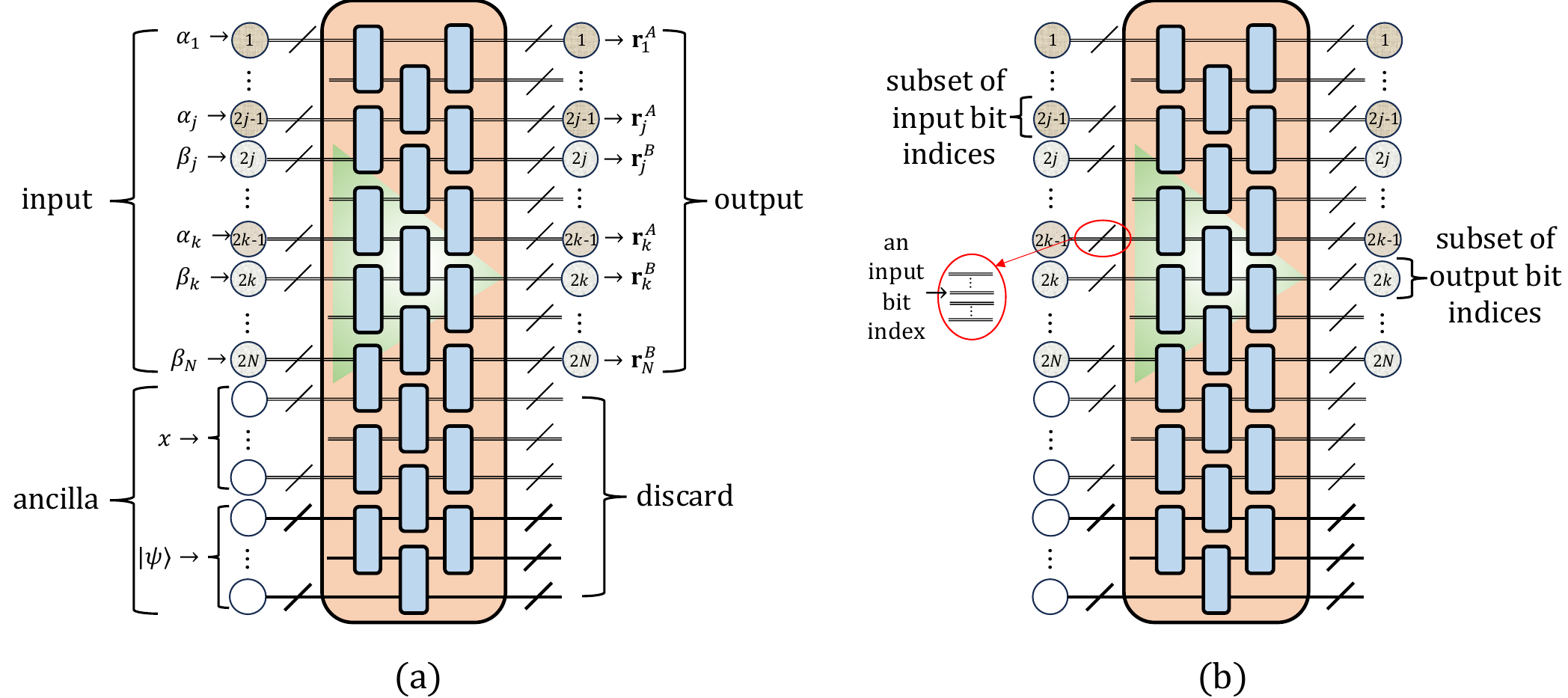}
\captionsetup{justification=Justified,singlelinecheck=false}
\caption{(a) A circuit to solve the second round of problem $R_N$. The input comprises classical bits $(\alpha_1,\beta_1,\cdots,\alpha_N,\beta_N)$. Note that $\alpha_i$ takes value from $\{\pm 1\}^6\times\left(\mathcal{Q}^A\cup\{\bot\}\right)$ and contains $\log(|\mathcal{Q}^A|+1)+6$ bits. The case is similar for $\beta_i$. Besides, the circuit also has an arbitrary number of bits with value $x$ and qubits $\ket{\psi}$ as ancillas. The values of ancillas are independent of the input information. The output of the relation problem comprises classical bits $(\mathbf{r}_1^A,\mathbf{r}_1^B,\cdots,\mathbf{r}_N^A,\mathbf{r}_N^B)$. Each $\mathbf{r}_i^A$ or $\mathbf{r}_i^B$ contains three bits. Other qubits and bits within the circuit are discarded. (b) Diagram to label a subset of the input bits, a subset of the output bits, and an input bit. In Lemma~\ref{lemma:hardinstance}, the subset of the input bits $I$ is the set that takes value in $\mathcal{Q}^A$ or $\mathcal{Q}^B$ and does not take value of $\bot$. The subset of the output bits $O$ is the set of bits outputting $\mathbf{r}_j^A$ or $\mathbf{r}_k^B$.
}
\label{fig:lightcone}
\end{figure*}

Using the idea of information not spreading a lot, we will show that with high probability, an input value $\alpha_j$ ($\beta_k$) is independent of the output $\mathbf{r}_k^B$ ($\mathbf{r}_j^A$), as presented in Lemma~\ref{lemma:hardinstance}. Before proving it, we first present the following lemma.

\begin{lemma}\label{lemma:fanininter}
Let $\mathcal{C}$ be a depth-$D$ circuit with classical inputs and an arbitrary number of quantum and classical ancillas, which comprises gates with fan-in upper bounded by $K$. The output of $\mathcal{C}$ is read out from the classical bits at the end of the circuit. Then, the following holds:

Let $O$ be a fixed subset of output bit indices, and suppose $I$ is a randomly chosen subset of input bit indices such that for any input bit index $v$,
\begin{equation}
\Pr_I[v\in I]\leq \eta.
\end{equation}
Then
\begin{equation}
\Pr_I[O\cap L_C^{\rightarrow}(I)\neq \emptyset]\leq \eta|O|2^{|O|}K^D.
\end{equation}
\end{lemma}

\begin{proof}
\begin{equation}
\begin{split}
\Pr_I[O\cap L_C^{\rightarrow}(I)\neq \emptyset] &= \sum_{P\subseteq O, P\neq \emptyset} \Pr_I[O\cap L_C^{\rightarrow}(I) = P]\\
&\leq \sum_{P\subseteq O, P\neq \emptyset} \Pr_I[I\cap L_C^{\leftarrow}(P) \neq \emptyset]\\
&\leq \sum_{P\subseteq O, P\neq \emptyset} \sum_{v\in L_C^{\leftarrow}(P)}\Pr_I[v\in I]\\
&\leq \sum_{P\subseteq O, P\neq \emptyset} \sum_{v\in L_C^{\leftarrow}(P)}\eta\\
&\leq \sum_{P\subseteq O, P\neq \emptyset} \abs{L_C^{\leftarrow}(P)}\eta\\
&\leq \sum_{P\subseteq O, P\neq \emptyset} \abs{P}K^D \eta\\
&\leq 2^{|O|}|O|K^D\eta.
\end{split}
\end{equation}

\end{proof}

\begin{lemma}\label{lemma:hardinstance}
Consider a depth-$D$ circuit composed of gates of fan-in at most $K$. The input $q=(\alpha_1,\beta_1,\cdots,\alpha_N,\beta_N)$ of the circuit of the second round is determined by a tuple $(j,k,\alpha,\beta)$ with $1\leq j < k\leq N$ and given by Eq.~\eqref{eq:shallowgameRound2}. We denote the set of all possible inputs as $S$. The output of the circuit of the second round is given by $r=(\mathbf{r}_1^A,\mathbf{r}_1^B,\cdots,\mathbf{r}_N^A,\mathbf{r}_N^B)$. Define the event $E_{\mathcal{C}}\subset S$ in which the input parameters satisfy
\begin{equation}\label{eq:EcInteractive}
\mathrm{supp}(\mathbf{r}_j^A) \cap L_C^{\rightarrow} (\mathrm{supp}(\beta_k)) = \emptyset\ \mathrm{and}\ \mathrm{supp}(\mathbf{r}_k^B) \cap L_C^{\rightarrow} (\mathrm{supp}(\alpha_j)) = \emptyset.
\end{equation}
Here, $\mathrm{supp}(x)$ means the bits carrying on the value $x$.
Under a uniform choice of input from $S$, the event $E_{\mathcal{C}}$ occurs with probability $\Pr[E_{\mathcal{C}}]\geq 1-\frac{48K^{D}}{N}$.
\end{lemma}
\begin{proof}
We consider a random input from the set $S$, which is constructed by a randomly generated tuple $(j,k,\alpha,\beta)$. Note that $j$ and $k$ are two different numbers uniformly and randomly picked from $\{1, 2, \cdots, N\}$ while $\alpha$ and $\beta$ are randomly and uniformly picked from $\mathcal{Q}^A$ and $\mathcal{Q}^B$, respectively. Consider the input bit set as $I = \mathrm{supp}(\alpha_j)$ as the set that takes value in $\mathcal{Q}^A$ and does not take value of $\bot$ in Lemma~\ref{lemma:fanininter} and suppose the input bit $v$ is located at site $j^*$, as shown in Fig.~\ref{fig:lightcone}(b).
We have
\begin{equation}
\Pr_{\mathrm{supp}(\alpha_j)}[v \in \mathrm{supp}(\alpha_j)] = \Pr_{1\leq j\leq N}[j = j^*] = \frac{1}{N}.
\end{equation}

Meanwhile, consider the subset of the output bits $O$ is the set of bits outputting $\mathbf{r}_k^B$, based on Lemma~\ref{lemma:fanininter}, we obtain
\begin{equation}
\Pr_{\mathrm{supp}(\alpha_j)}[\mathrm{supp}(\mathbf{r}_k^B) \cap L_C^{\rightarrow} (\mathrm{supp}(\alpha_j)) \neq \emptyset]\leq 2^{|\mathbf{r}_k^B|}|\mathbf{r}_k^B|K^D\frac{1}{N}.
\end{equation}
Note that $\mathbf{r}_k^B$ only has 3 output bits, then
\begin{equation}
\Pr_{\mathrm{supp}(\alpha_j)}[\mathrm{supp}(\mathbf{r}_k^B) \cap L_C^{\rightarrow} (\mathrm{supp}(\alpha_j)) \neq \emptyset]\leq \frac{24K^D}{N}.
\end{equation}
Similarly, we get
\begin{equation}
\Pr_{\mathrm{supp}(\beta_k)}[\mathrm{supp}(\mathbf{r}_j^A) \cap L_C^{\rightarrow} (\mathrm{supp}(\beta_k)) \neq \emptyset]\leq \frac{24K^D}{N}.
\end{equation}
Thus,
\begin{equation}
\begin{split}
\Pr_q[E_{\mathcal{C}}] &= \Pr_{\mathrm{supp}(\alpha_j),\mathrm{supp}(\beta_k)}[\mathbf{r}_j^A \cap L_C^{\rightarrow} (\beta_k) = \emptyset \cap \mathbf{r}_k^B \cap L_C^{\rightarrow} (\alpha_j) = \emptyset]\\
&= 1 - \Pr_{\mathrm{supp}(\alpha_j),\mathrm{supp}(\beta_k)}[\mathbf{r}_j^A \cap L_C^{\rightarrow} (\beta_k) \neq \emptyset \cup \mathbf{r}_k^B \cap L_C^{\rightarrow} (\alpha_j) \neq \emptyset]\\
&\geq 1 - \Pr_{\mathrm{supp}(\beta_k)}[\mathbf{r}_j^A \cap L_C^{\rightarrow} (\beta_k) \neq \emptyset] - \Pr_{\mathrm{supp}(\alpha_j)}[\mathbf{r}_k^B \cap L_C^{\rightarrow} (\alpha_j) \neq \emptyset]\\
&\geq 1 - \frac{48K^D}{N}.
\end{split}
\end{equation}

\end{proof}

With Lemma~\ref{lemma:hardinstance}, we are able to achieve our ultimate goal to prove the hardness of $R_N$ for magic-free shallow circuits. The main idea is that with high probability, the event defined in Lemma~\ref{lemma:hardinstance} would happen, and if this event happens, the magic-free circuit cannot give a correct output in the second round as the nonlocal game requires magic to win with certainty.

\begin{theorem}
Let $\mathcal{C}$ be a depth-$D$ circuit with classical input values and classical and quantum ancillas, which only comprises magic-free operations with fan-in upper bounded by $K$. Now, consider the classical input $q=(\alpha_1,\beta_1,\cdots,\alpha_N,\beta_N)$ of the second round determined by Eq.~\eqref{eq:shallowgameRound2} with $\alpha$ and $\beta$ selected uniformly at random from $\mathcal{Q}^A$ and $\mathcal{Q}^B$. Then the average probability that $\mathcal{C}$ outputs $r=(\mathbf{r}_1^A,\mathbf{r}_1^B,\cdots,\mathbf{r}_N^A,\mathbf{r}_N^B)$ such that $r$ and $q$ satisfy Eq.~\eqref{eq:successrelationRound2} is at most $p_{\mathrm{Clif}} + \frac{48K^D}{N}$, with $p_{\mathrm{Clif}}$ given in Lemma~\ref{lemma:modifiedwinprob}. To meet the requirements with a success probability larger than $(1+p_{\mathrm{Clif}})/2$, the circuit depth requirement is $\Theta(\log N)$.
\end{theorem}
\begin{proof}
When the circuit $\mathcal{C}$ successfully completes the second round of the relation problem $R_N$, i.e., $r$ and $q$ satisfy Eq.~\eqref{eq:successrelationRound2}, the success condition is making the inputs and outputs $\mathbf{r}_j^A,\mathbf{r}_k^B$ of the second round satisfy the relation defined by the BCS game, i.e., Eq.~\eqref{eq:successrelationRound2}. Also, when the event $E_{\mathcal{C}}$ (cf. Eq.~\eqref{eq:EcInteractive}) happens, the output $\mathbf{r}_j^A$ only depends on $\alpha_j$ and does not depend on $\beta_k$. And the reverse is true for $\mathbf{r}_k^B$. It reduces to the case that Alice and Bob are trying to win the BCS non-local game without classical communication. As the circuit only comprises $\ket{0}$ as input quantum states, Clifford gates, and Pauli measurements, it means that Alice and Bob need to win this non-local game with Pauli measurements and magic-free states, whose winning probability is upper-bounded by $p_{\mathrm{Clif}}$. That is, $\Pr_q[\mathrm{success}|E_{\mathcal{C}}]\leq p_{\mathrm{Clif}}$. Then, we conclude that the average success probability of $\mathcal{C}$ to output a correct relation between $r$ and $q$ of the second round is 
\begin{equation}
    \Pr_q[\mathrm{success}]
    \leq \Pr_q[\mathrm{success} | E_{\mathcal{C}}] + (1-\Pr_q[E_{\mathcal{C}}]) \leq p_{\mathrm{Clif}} + \frac{48K^D}{N}.
\end{equation}
To output the correct relation with a success probability larger than $(1+p_{\mathrm{Clif}})/2$, the circuit depth $D$ of the second round has a lower bound as below.
\begin{equation}
    p_{\mathrm{Clif}}+\frac{48K^D}{N}\geq\frac{1+p_{\mathrm{Clif}}}{2}
    \ \Leftrightarrow\ 
    D\geq\frac{\log N+\log\frac{1-p_{\mathrm{Clif}}}{96}}{\log K}=\Omega(\log N).
\end{equation}
On the other hand, as stated in the main text, there is a classical circuit with circuit depth $O(\log N)$ that solves the problem. Therefore, the bound on the circuit depth for magic-free circuits to solve the problem is tight. This finishes the proof.

\end{proof}

\subsection{Sampling problem}\label{app:sampling}
Here, we introduce the second task -- the sampling problem, which is also labeled by $R_N$ with $N$ representing the problem size. Different from the previous part, here the task only has one round. Similarly, one can assume that two players, Alice and Bob, collaborate with each other to solve $R_N$. The input and output of $R_N$ are given below.
\begin{enumerate}
\item Input: Alice and Bob are given a question,
\begin{equation}\label{eq:Rninput}
q=(\alpha_1,\beta_1,\cdots,\alpha_N,\beta_N),\alpha_i\in\mathcal{Q}^A\cup\{\perp\},\beta_i\in\mathcal{Q}^B\cup\{\perp\},
\end{equation}
where $q$ stands for ``question'', $\alpha_i$ is the input on Alice's side at site $i\in[N]$, and $\beta_i$ is the input on Bob's side at site $i\in[N]$. $\mathcal{Q}^A$ consists of the set of constraints in the BCS, and $\mathcal{Q}^B$ consists of the set of variables in the BCS. Here, $\perp$ represents a null input.
\item Output: Alice and Bob need to return a reaction to the question,
\begin{equation}\label{eq:Rnoutput}
r=(\mathbf{r}_1^A,\mathbf{r}_1^B,\cdots,\mathbf{r}_N^A,\mathbf{r}_N^B),
\end{equation}
where $r$ stands for ``reaction'', $\mathbf{r}_i^A=(r_i^A(1),r_i^A(2),r_i^A(3))\in\{\pm1\}^3$ is the output on Alice's side at site $i$, and similarly on Bob's side.
\end{enumerate}

Now, we define the problem $R_N$. In the computation task, Alice and Bob are promised to receive an instance given by a tuple $(j,k,\alpha,\beta),1\leq j<k\leq N$, which defines the input as
\begin{equation}\label{eq:shallowgame}
\alpha_i=\begin{cases}
\alpha, & \text{if $i=j$}, \\
\perp, & \text{if $i\neq j$},
\end{cases}
\quad
\beta_i=\begin{cases}
\beta, & \text{if $i=k$}, \\
\perp, & \text{if $i\neq k$}.
\end{cases}
\end{equation}
That is, we require the sites $j$ on Alice's side and $k$ on Bob's side to play the BCS nonlocal game with questions $\alpha$ and $\beta$, respectively. Alice and Bob are required to give an output satisfying either of the following requirements:

\begin{enumerate}[label={\textbf{Case \arabic*}}]
\item\label{enum:shallowgame1} For any $l \in \{1,2,3\}$,
\begin{equation}\label{eq:correctBSM}
\begin{split}
\prod_{j<i\leq k}r_i^A(l)=+1, \\
\prod_{j\leq i<k}r_i^B(l)=+1, \\
\end{split}
\end{equation}
and
\begin{equation}\label{eq:successrelation}
(\mathbf{r}_j^A,\mathbf{r}_k^B)=f(\alpha,\beta),
\end{equation}
where $f$ is the relation defined by the BCS nonlocal game.
\item\label{enum:shallowgame2} There exists $l\in\{1,2,3\}$ such that,
\begin{equation}
\prod_{j<i\leq k}r_i^A(l)=-1,
\end{equation}
or
\begin{equation}
\prod_{j\leq i < k}r_i^B(l)=-1.
\end{equation}
\end{enumerate}
In addition, we require that for any question $q$, \ref{enum:shallowgame1} occurs with a probability no smaller than a positive constant value $p\in(0,1/64]$.

Below, we show that the sampling problem $R_N$ can be solved by a \QNCzero circuit but cannot be solved by any (fixed) \ClifNCzero circuit. We first show that a shallow circuit with generic bounded fan-in quantum gates exists that perfectly completes this task. The idea is the same as that in the previous part. Alice and Bob first share perfect EPR pairs via entanglement swapping, and then play the nonlocal game on the shared EPR pairs:
\begin{enumerate}
\item Alice and Bob share $3N$ pairs of EPR states, $\ket{\Phi^+}^{\otimes 3N}$, where $\ket{\Phi^+}=(\ket{00}+\ket{11})/\sqrt{2}$, and arrange them in three layers, denoted by $q_{2i-1}(l),q_{2i}(l),i\in[N],l\in\{1,2,3\}$, where qubits $q_{2i-1}(l)$ and $q_{2i}(l)$ reside in the state $\ket{\Phi^+}$. Alice holds the qubits $q_{2i-1}(l)$ and Bob holds the qubits $q_{2i}(l)$.
\item For any $j\leq i\leq k-1$, perform an entanglement swapping operation between pairs of EPRs with a BSM on qubits $q_{2i}(l)$ and $q_{2i+1}(l)$. Denote the Bell state measurement results on the pair of adjacent qubits $q_{2i}(l)$ and $q_{2i+1}(l)$ as $r_i^B(l)$ and $r_{i+1}^A(l)$, respectively.
\item On the three pairs of qubits $q_{2j-1}(l)$ and $q_{2k}(l)$, Alice and Bob perform the measurements corresponding to the winning strategy in the BCS nonlocal game and obtain outputs $(r_j^A(l),r_k^B(l))$.
\item Take an arbitrary measurement on the qubits that are not measured and record the measurement results with respect to the site indices.
\end{enumerate}

Note that at the end of the entanglement swapping operations, qubits $q_{2j-1}(l)$ and $q_{2k}(l)$ reside in $\ket{\Phi^+}$ with probability $1/4$ for each $l\in \{1, 2, 3\}$. By construction, this strategy naturally meets the problem requirements, and the first requirement defined through Eqs.~\eqref{eq:correctBSM} and~\eqref{eq:successrelation} is met with probability $1/64$. Similarly, this strategy can be done with 3-bounded fan-in gates in a constant quantum circuit.

In the following, we prove the hardness of $R_N$ for magic-free shallow circuits. Similar to Lemma~\ref{lemma:hardinstance}, we first show that with high probability, output $\mathbf{r}_j^A$ is independent of input $\beta_k$, and output $\mathbf{r}_k^B$ is independent of input $\alpha_j$. The result is summarized in the following lemma.

\begin{lemma}\label{lemma:hardinstancesampling}
Consider a depth-$D$ circuit composed of gates of fan-in at most $K$. The input of the circuit $q=(\alpha_1,\beta_1,\cdots,\alpha_N,\beta_N)$ is determined by a tuple $(j,k,\alpha,\beta)$ with $1\leq j < k\leq N$ and given by Eq.~\eqref{eq:shallowgame}. We denote the set of all possible inputs as $S$. The output of the circuit is given by $r=(\mathbf{r}_1^A,\mathbf{r}_1^B,\cdots,\mathbf{r}_N^A,\mathbf{r}_N^B)$. Define the event $E_{\mathcal{C}}\subset S$ in which the input parameters satisfy
\begin{equation}
\mathrm{supp}(\mathbf{r}_j^A) \cap L_C^{\rightarrow} (\mathrm{supp}(\beta_k)) = \emptyset\ \mathrm{and}\ \mathrm{supp}(\mathbf{r}_k^B) \cap L_C^{\rightarrow} (\mathrm{supp}(\alpha_j)) = \emptyset.
\end{equation}
Here, $\mathrm{supp}(x)$ means the bits carrying on the value $x$.
Under a uniform choice of input from $S$, the event $E_{\mathcal{C}}$ occurs with probability $\Pr[E_{\mathcal{C}}]\geq 1-\frac{48K^{D}}{N}$.
\end{lemma}

The proof of Lemma~\ref{lemma:hardinstancesampling} is the same as that of Lemma~\ref{lemma:hardinstance}, and we omit it here. As a consequence, the event defined in Lemma~\ref{lemma:hardinstancesampling} would happen with high probability. Based on Lemma~\ref{lemma:hardinstancesampling}, we prove the following theorem to show that magic-free shallow circuits cannot solve the sampling problem $R_N$ prefectly.

\begin{theorem}
Let $\mathcal{C}$ be a (fixed) depth-$D$ circuit with classical input values and classical and quantum ancillas, which only comprises magic-free operations with fan-in upper bounded by $K$. Now, consider the classical input $q=(\alpha_1,\beta_1,\cdots,\alpha_N,\beta_N)$ determined by Eq.~\eqref{eq:shallowgame} with $\alpha_j$ and $\beta_k$ selected uniformly at random from $\mathcal{Q}^A$ and $\mathcal{Q}^B$. Then for any constant value $p\in (0,1)$, the average probability that $\mathcal{C}$ outputs $r=(\mathbf{r}_1^A,\mathbf{r}_1^B,\cdots,\mathbf{r}_N^A,\mathbf{r}_N^B)$ such that
\begin{enumerate}
\item[(1)] $r$ and $q$ satisfy the requirements in \ref{enum:shallowgame1} and~\ref{enum:shallowgame2},
\item[(2)] $\forall q$, outputs \ref{enum:shallowgame1} with probability no smaller than $p$,
\end{enumerate}
is at most $1-p(1-\frac{48K^D}{N}-p_{\mathrm{Clif}})$, with $p_{\mathrm{Clif}}$ given in Lemma~\ref{lemma:modifiedwinprob}. To meet the requirements with a success probability larger than $1-p(1-p_{\mathrm{Clif}})/2$, the circuit depth requirement is $\Theta(\log N)$.
\end{theorem}
\begin{proof}
Given a specific circuit, $\mathcal{C}$, the average success probability of $\mathcal{C}$ to output a correct relation between $r$ and $q$ is
\begin{equation}\label{eq:circuitsuccess}
\begin{split}
\Pr_q[\mathrm{success}] \leq& \Pr_q[\ref{enum:shallowgame1}, E_{\mathcal{C}}]\Pr_q[\mathrm{success} | \ref{enum:shallowgame1}, E_{\mathcal{C}}]
+1-\Pr_q[\ref{enum:shallowgame1}, E_{\mathcal{C}}]\\
=& 1 - \Pr_q[\ref{enum:shallowgame1}, E_{\mathcal{C}}]+\Pr_q[\mathrm{success},\ref{enum:shallowgame1}, E_{\mathcal{C}}]\\
=& 1 - \sum_q\Pr[\ref{enum:shallowgame1}, E_{\mathcal{C}},q](1-\Pr[\mathrm{success} | \ref{enum:shallowgame1},E_{\mathcal{C}},q])\\
=& 1 - \sum_q \Pr[\ref{enum:shallowgame1}|E_{\mathcal{C}},q]\Pr[E_{\mathcal{C}}|q]\Pr[q](1-\Pr[\mathrm{success} | \ref{enum:shallowgame1},E_{\mathcal{C}},q])\\
\leq& 1 - p\sum_q\Pr[E_{\mathcal{C}}|q]\Pr[q](1-\Pr[\mathrm{success} | \ref{enum:shallowgame1},E_{\mathcal{C}},q])\\
=& 1 - p\sum_{j,k,\alpha,\beta}\Pr[E_{\mathcal{C}}|j,k,\alpha,\beta]\Pr[j,k,\alpha,\beta](1-\Pr[\mathrm{success} | \ref{enum:shallowgame1},E_{\mathcal{C}},j,k,\alpha,\beta])\\
=& 1 - p\sum_{j,k,\alpha,\beta}\Pr[E_{\mathcal{C}}|j,k]\Pr[j,k]\Pr[\alpha,\beta](1-\Pr[\mathrm{success} | \ref{enum:shallowgame1},E_{\mathcal{C}},j,k,\alpha,\beta])\\
=& 1 - p\sum_{j,k}\Pr[E_{\mathcal{C}}|j,k]\Pr[j,k] \sum_{\alpha,\beta}\Pr[\alpha,\beta](1-\Pr[\mathrm{success} | \ref{enum:shallowgame1},E_{\mathcal{C}},j,k,\alpha,\beta])\\
\leq& 1-p(1-\frac{48K^D}{N})(1-p_{\mathrm{Clif}})\\
\leq& 1-p(1-\frac{48K^D}{N}-p_{\mathrm{Clif}})\\
\end{split}
\end{equation}
Here, $E_{\mathcal{C}}$ is the event defined in Lemma~\ref{lemma:hardinstancesampling}. In the fifth line, we use the condition (2), $\forall q, \Pr[\ref{enum:shallowgame1}|q] \geq p$, which implies $\Pr[\ref{enum:shallowgame1}|E_{\mathcal{C}}, q] \geq p$. In the seventh line, we use the independent condition of the probability distribution of the question, $\Pr[j,k,\alpha,\beta] = \Pr[j,k]\Pr[\alpha,\beta]$. In the ninth line, we use the following conditions.
\begin{align}
\label{eq:lightcone}\sum_{j,k}\Pr[j,k]\Pr[E_{\mathcal{C}}|j,k] = \Pr_q[E_{\mathcal{C}}] &\geq 1 - \frac{48K^D}{N};\\
\label{eq:successbound}\forall (j,k), \sum_{\alpha,\beta}\Pr[\alpha,\beta]\Pr[\mathrm{success} | \ref{enum:shallowgame1},E_{\mathcal{C}},j,k,\alpha,\beta]) &\leq p_{\mathrm{Clif}}.
\end{align}

Below, we prove the above two conditions. Based on the proof of Lemma~\ref{lemma:hardinstance} and Lemma~\ref{lemma:hardinstancesampling}, one can see that only the location $(j, k)$ determines whether the event $E_{\mathcal{C}}$ happens, and the tuple $(\alpha, \beta)$ does not have any influence. Thus, we have $\Pr[E_{\mathcal{C}}|j,k,\alpha,\beta] = \Pr[E_{\mathcal{C}}|j,k]$. Then,
\begin{equation}
\begin{split}
\sum_{j,k}\Pr[j,k]\Pr[E_{\mathcal{C}}|j,k] &= \sum_{\alpha,\beta}\Pr[\alpha,\beta] \sum_{j,k}\Pr[j,k]\Pr[E_{\mathcal{C}}|j,k,\alpha,\beta]\\
&=\sum_{j,k,\alpha,\beta}\Pr[j,k,\alpha,\beta] \Pr[E_{\mathcal{C}}|j,k,\alpha,\beta]\\
&=\Pr_q[E_{\mathcal{C}}]\\
&\geq 1 - \frac{48K^D}{N},
\end{split}
\end{equation}
which proves Eq.~\eqref{eq:lightcone}.

For Eq.~\eqref{eq:successbound}, note that when \ref{enum:shallowgame1} happens, the success condition is making the inputs $\alpha_j$ and $\beta_k$ and the outputs $\mathbf{r}_j^A$ and $\mathbf{r}_k^B$ satisfy the relation defined by the BCS game, i.e., Eq.~\eqref{eq:successrelation}. Also, when $E_{\mathcal{C}}$ happens, the output $\mathbf{r}_j^A$ only depends on $\alpha_j$ and does not depend on $\beta_k$. And the reverse is true for $\mathbf{r}_k^B$. Also, notice that the input distribution of $(\alpha,\beta)$ is still uniform. Eq.~\eqref{eq:successbound} reduces to the case that Alice and Bob are trying to win the BCS non-local game without classical communication under the uniform distribution of the input. As the circuit only comprises $\ket{0}$ as input quantum states, Clifford gates, and Pauli measurements, it means that Alice and Bob need to win this non-local game with Pauli measurements and magic-free states, whose winning probability is upper-bounded by $p_{\mathrm{Clif}}$. That concludes Eq.~\eqref{eq:successbound}.

Based on Eq.~\eqref{eq:circuitsuccess}, to output the correct relation with a success probability larger than $1-p(1-p_{\mathrm{Clif}})/2$, the circuit depth $D$ has a lower bound as below.
\begin{equation}
1-p(1-\frac{48K^D}{N}-p_{\mathrm{Clif}})\geq 1-p(1-p_{\mathrm{Clif}})/2 \Leftrightarrow D\geq \frac{\log N + \log\frac{1-p_{\mathrm{Clif}}}{96}}{\log K} = \Omega(\log N).
\end{equation}
On the other hand, as stated in the main text, there is a (fixed) classical circuit with circuit depth $O(\log N)$ that solves the problem. Therefore, the bound on the circuit depth for magic-free circuits to solve the problem is tight. This finishes the proof.

\end{proof}

Before ending this subsection, we comment on why the fixed condition of the circuit is required to prove the hardness. When a probabilistic circuit is allowed, the success probability will become
\begin{equation}\label{eq:successbound2}
\begin{split}
\Pr_q[\mathrm{success}] =& \sum_{\mathcal{C}} \Pr[\mathcal{C}]\Pr_q[\mathrm{success}|\mathcal{C}] \\
\leq& 1 - \sum_q\Pr[\mathcal{C}]\Pr[q]\Pr[\ref{enum:shallowgame1}| E_{\mathcal{C}},q,\mathcal{C}] \Pr[E_{\mathcal{C}}|q,\mathcal{C}](1-\Pr[\mathrm{success} | \ref{enum:shallowgame1},E_{\mathcal{C}},q,\mathcal{C}]).
\end{split}
\end{equation}
The term $\Pr[\mathcal{C}]$ represents the probability distribution of the circuit. In this case, the condition (2) becomes
\begin{equation}\label{eq:case1prob}
\sum_{\mathcal{C}}\Pr[\ref{enum:shallowgame1}|q,\mathcal{C}]\geq p
\end{equation}
When the randomness is available, one may not ascertain that for any $q$ and $\mathcal{C}$, $\Pr[\ref{enum:shallowgame1}|q,\mathcal{C}]\geq p$. For instance, given a circuit, \ref{enum:shallowgame1} only appears for specific questions $q$, and for other circuits, \ref{enum:shallowgame1} appears for other questions $q$. In this case, Eq.~\eqref{eq:case1prob} still holds, but now it is unable to remove term $\Pr[\ref{enum:shallowgame1}| E_{\mathcal{C}},q,\mathcal{C}]$ in Eq.~\eqref{eq:successbound2} and bound the success probability. In fact, the magic-free circuit might post-select question $q$ by \ref{enum:shallowgame1}. Thus, when \ref{enum:shallowgame1} happens, only part of the questions is input into the circuit, and the magic-free circuit might solve the sampling problem by only considering these questions. Nonetheless, this strategy only works when $\Pr[\ref{enum:shallowgame1}|q,\mathcal{C}]\geq p$ fails. If every circuit in the probabilistic circuit satisfies $\Pr[\ref{enum:shallowgame1}|q,\mathcal{C}]\geq p$, the success probability can still be bounded.

\section{Finding Potential Magic-Necessary Linear Binary Constraint Systems}\label{supp:GeneralProcedure}
In this section, we discuss finding other instances of linear BCS that require magic for a perfect quantum solution. It is challenging to develop a general procedure for this target. Instead, we provide a ``guess-and-check'' procedure: (1) First, obtain a potential BCS, and (2) Second, verify whether the BCS has a solution over the Pauli group. For the first step, one can use the group embedding results in Ref.~\cite{slofstra2019set}. Building on the group-theoretic results, including Ref.~\cite{cleve2014characterization,cleve2017perfect}, Ref.~\cite{slofstra2019set} provides an efficient procedure to embed a group into a BCS, which is a group homomorphism of the original group to a non-trivial BCS solution group; hence the procedure constructs a BCS that necessarily has a (quantum) solution. In particular, the output solution group inherits the representation properties of the original group. However, such a BCS may have a classical solution or a quantum solution over the Pauli group, where magic is absent. We develop an efficient classical algorithm for this issue to decide whether a linear BCS has a Pauli-string solution. Besides aiding the search for non-trivial linear BCS, we hope this result can help explore the decidability problems for general BCS and nonlocal games~\cite{fu2021membership}.

\subsection{Slofstra's group embedding procedure}\label{Supp:Slofstra}
A group $G$ is said to be embedded into group $K$ if there exists an injective group homomorphism $\phi:G\rightarrow K$. One can pose additional requirements to the group embedding to guarantee the inheritance of group representation properties; see Definitions 10 and 14 in Ref.~\cite{slofstra2019set} for example. This is also one of the core issues in the group embedding procedure in Ref.~\cite{slofstra2019set}. For our purpose of finding magic-necessary BCS, some analysis on the group representation inheritance issue may be redundant. Nevertheless, we faithfully review the group embedding results of Ref.~\cite{slofstra2019set} here and leave the problem of simplifying the procedure to future work. For the convenience of stating the group embedding results, we use the Boolean variables and parity constraints instead of the sign variables and multilinear constraints in this subsection. The conversion between sign variables and Boolean variables is given in Eq.~\eqref{eq:signtobool}.

As we now use the Boolean variable notation, a BCS can be compactly written as $M\vec{\mathbf{v}}=\vec{\mathbf{c}}$, where $M$ is an $m\times n$ Boolean matrix, $\vec{\mathbf{v}}=(v_1,\cdots, v_n)^\mathrm{T}$ is the vector of variables, and $\vec{\mathbf{c}}=(c_1,\cdots,c_m)^\mathrm{T}$ is the vector of constraints. The non-zero elements in the $j$'th row of $M$ define the set of variables presented in the $j$'th constraint, $\mathcal{S}_j$. We first define several classes of groups, including a restatement of the BCS solution group with the current notations.

\begin{definition}[Solution group of a linear BCS using Boolean variables]
Given a linear BCS with $n$ binary variables $\{v_i\}$ and $m$ constraints $\{c_j\}$ specified by $M\vec{\mathbf{v}}=\vec{\mathbf{c}}$, the solution group of the BCS is defined as the group with the following presentation:
\begin{equation}
\begin{split}
\Gamma(M,\vec{\mathbf{c}})=\{\{J,g_i:i\in[n]\}:\{g_i^2&=e,\forall i\in[n], \\
J^2&=e, \\
\forall j\in[m], g_kg_l&=g_lg_k, \forall k,l\in\mathcal{S}_j, \\
Jg_i&=g_iJ,\forall i\in[n], \\
\prod_{i\in\mathcal{S}_j}g_i&=J^{c_j},\forall j\in[m]\}\},
\end{split}
\end{equation}
where $e$ defines the identity operator of the group. We take the convention that $J^0=e$.
\end{definition}

\begin{definition}[Linear-plus-conjugacy group]
Given a linear BCS with $n$ binary variables $\{v_i\}$ and $m$ constraints $\{c_j\}$ specified by $M\vec{\mathbf{v}}=\vec{\mathbf{c}}$, and $\mathcal{C}\subseteq[n]\times[n]\times[n]$ with $[n]=\{1,\cdots,n\}$, the linear-plus-conjugacy group $\Gamma(M,\vec{\mathbf{c}},\mathcal{C})$ is defined as
\begin{equation}
\Gamma(M,\vec{\mathbf{c}},\mathcal{C})\equiv\langle\Gamma(M,\vec{\mathbf{c}}):g_ig_jg_i=g_k,\forall (i,j,k)\in\mathcal{C}\rangle,
\end{equation}
where the relators $g_ig_jg_i=g_k$ are additionally posed to the solution group $\Gamma(M,\vec{\mathbf{c}})$.
\end{definition}

\begin{definition}[Homogeneous linear-plus-conjugacy group]
Given an $m\times n$ Boolean matrix $M$ where the set of non-zero elements in the $j$'th row is given by $\mathcal{S}_j$, and $\mathcal{C}\subseteq[n]\times[n]\times[n]$, the homogeneous linear-plus-conjugacy group $\Gamma_0(M,\mathcal{C})$ is defined as
\begin{equation}
\begin{split}
\Gamma_0(M,\mathcal{C})=\{\{g_i:i\in[n]\}:\{g_i^2&=e,\forall i\in[n], \\
\forall j\in[m], g_kg_l&=g_lg_k, \forall k,l\in\mathcal{S}_j, \\
\prod_{i\in\mathcal{S}_j}g_i&=e,\forall j\in[m], \\
g_ig_jg_i&=g_k,\forall(i,j,k)\in\mathcal{C}\}\}.
\end{split}
\end{equation}
\end{definition}

\begin{definition}[Extended homogeneous linear-plus-conjugacy group]
Given an $m\times n$ Boolean matrix $M$, $\mathcal{C}_0\subseteq[n]\times[n]\times[n]$, $\mathcal{C}_1\subseteq[l]\times[n]\times[n]$, and $L$ an $l\times l$ lower-triangular matrix with non-negative integer entries, the extended homogeneous linear-plus-conjugacy group $E\Gamma_0(M,\mathcal{C}_0,\mathcal{C}_1,L)$ is defined as
\begin{equation}
\begin{split}
E\Gamma_0(M,\mathcal{C}_0,\mathcal{C}_1,L)\equiv\langle\Gamma_0(M,\mathcal{C}_0),h_1,\cdots,h_l:
h_ig_jh_i^{-1}&=g_k,\forall(i,j,k)\in\mathcal{C}_1, \\
h_ih_jh_i^{-1}&=h_j^{L_{ij}},\forall i>j \wedge L_{ij}>0\rangle,
\end{split}
\end{equation}
where $L_{ij}$ refers to the element on the $i$'th row and $j$'column of matrix $L$.
\label{Def:EGroup}
\end{definition}

With the above definition, Ref.~\cite{slofstra2019set} proves the following embedding results:
\begin{theorem}[\cite{slofstra2019set}]
Suppose a group $G$ has a presentation in the form of an extended homogeneous linear-plus-conjugacy group given by $E\Gamma_0(M,\mathcal{C}_0,\mathcal{C}_1,L)$. Then there are the following group embedding results:
\begin{enumerate}
\item There exists a group embedding of $E\Gamma_0(M,\mathcal{C}_0,\mathcal{C}_1,L)$ into a homogeneous linear-plus-conjugacy group (Proposition 33 in Ref.~\cite{slofstra2019set}):
\begin{equation}
\begin{split}
E\Gamma_0(M,\mathcal{C}_0,\mathcal{C}_1,L)&\rightarrow\Gamma_0(M',\mathcal{C}), \\
\end{split}
\end{equation}
where $\Gamma_0(M',\mathcal{C})$ is a homogeneous linear-plus-conjugacy group. 

\item The extended $\Gamma_0(M',\mathcal{C})$ can be transformed into a linear-plus-conjugacy group (see the remark after Definition 31 in Ref.~\cite{slofstra2019set}):
\begin{equation}
\Gamma_0(M',\mathcal{C})\times\mathbb{Z}_2=\Gamma(M',0,\mathcal{C}).
\end{equation}

\item By adding relations with respect to one group element $J\in\Gamma(M',0,\mathcal{C}),J\neq e$, which extends the matrix $M'$ into $N$ and adds non-homogeneous linear constraints that involve $J$, extend the linear-plus-conjugacy group:
\begin{equation}
\Gamma(M',0,\mathcal{C})\rightarrow\Gamma(N,\vec{\mathbf{c}},\mathcal{C}'),
\end{equation}
where $\vec{\mathbf{c}}\neq 0$, with some entries equal to $J$.

\item There exists a group embedding of $\Gamma(N,\vec{\mathbf{c}},\mathcal{C}')$ into a linear group (Proposition 27 in Ref.~\cite{slofstra2019set}):
\begin{equation}
\Gamma(N,\vec{\mathbf{c}},\mathcal{C})\rightarrow\Gamma(N',\vec{\mathbf{c}}'),
\end{equation}
which defines a BCS solution group that has a non-trivial group element $J\neq e$.
\end{enumerate}
\end{theorem}

The proof is constructive, thus one can derive the concrete groups in each step. In brief, as long as a group can be presented in the form of Def.~\ref{Def:EGroup}, which is an extended homogeneous linear-plus-conjugacy group, it can be converted into a BCS solution group after a series of group embeddings and a proper construction of non-trivial relators with respect to a group element $J\neq e$, which is set to correspond to $-1$ in the BCS. As promised by Lemma~\ref{lemma:nontrivialsolution}, the underlying BCS of the solution group has an (operator-valued) solution.

\subsection{Efficient algorithm for finding perfect Pauli-string solutions to linear BCS}\label{algo}
Next, we provide an efficient algorithm to determine the existence of a Pauli-string solution to a general linear BCS. Combined with Slofstra's group embedding procedure, one can guess and test BCS instances to search for a potential BCS that has a non-trivial solution group other than the Pauli group. 

We first present some fundamental properties of Pauli-string observables.

\begin{lemma}\label{lemma:commutatorprop}
Suppose $A_1,A_2,\cdots,A_n$ are Pauli-string observables. For $i,j\in[n]$, define $C_{ij}=A_iA_jA_iA_j$ as the commutator between $A_i$ and $A_j$. Then, $C_{ij}$'s have the following properties:
\begin{enumerate}
\item $C_{ij}\in\{\pm\mathbb{I}\}$. Specifically, $C_{ij}=\mathbb{I}$ when $A_iA_j-A_jA_i=0$, and $C_{ij}=-\mathbb{I}$ when $A_iA_j+A_jA_i=0$;
\item $C_{ij}=C_{ji}$ and $C_{ii}=\mathbb{I}$;
\item $A_iA_j=C_{ij}A_jA_i$.
\end{enumerate}
\end{lemma}

The proof of Lemma~\ref{lemma:commutatorprop} is straightforward, and we leave it to the readers as an exercise. According to the first property, $C_{ij}$ is always proportional to $\mathbb{I}$ and thus commute with all the $A_i$'s. We can treat $C_{ij}$'s as numbers $\pm 1$ for simplicity. The second property shows that for a group of $n$ $A_i$'s, the number of independent commutators $C_{ij}$'s among them is at most $n(n-1)/2$. The third property shall be vital for our later discussions, as it allows us to swap two adjacent variables $A_i$ and $A_j$ in a product of Pauli strings up to an additional coefficient $C_{ij}$. Together with the first property, for a product of Pauli-string variables, we can arbitrarily rearrange their order up to a change in the sign.

Given a linear BCS, we first determine if it has a classical solution, i.e., $A_i\in\{\pm 1\}$ for all $i\in[n]$. This is equivalent to solving a system of linear equations over $\mathbb{Z}_2$, which can be done in $\mathrm{poly}(n)$ steps through, for example, the Gaussian elimination method.
Going back to determine if the BCS has a Pauli-string solution, if we hope to apply a similar procedure, the only obstacle is that the variables might not commute with each other. Nevertheless, thanks to the nice properties of Pauli strings in Lemma~\ref{lemma:commutatorprop}, we can do the same thing as finding a classical solution with at most a difference in sign, which we record as a sign variable $C_i$. 
Should the BCS have a Pauli-string solution, at the end of the elimination, we can express each variable $A_i$ as a product of some variables, which we term the ``free variables,'' multiplied by a plus or minus sign $C_i\in\{\pm 1\}$.
We use the terminology ``free variables'' as they are allowed to take any value, while the remaining variables depend on their values.  
Now we give the rigorous statement and prove it.

\begin{lemma}\label{lemma:solvelineareqs}
For a linear BCS with $n$ variables $A_1,\cdots,A_n$ and $m$ constraints, if it has a Pauli-string solution, then there exists a set of free variables $\{A_{i_k}\}_k$, such that each variable in the BCS can be represented in the form of $A_i=C_iA_{i_1}A_{i_2}\cdots A_{i_k}\cdots$, where $C_i\in\{\pm 1\}$ and $A_{i_k}$'s are arranged with the subscript $k$ from small to large. This result can be obtained in $\mathrm{poly}(n,m)$ steps.
\end{lemma}

\begin{proof}
When finding an operator-valued solution to a linear BCS, we require the variables in the same constraint to be compatible, as discussed in Appendix~\ref{supp:BCSPre}. Now we prove a stronger statement that does not rely on this requirement.

We prove the lemma by mathematical induction on $m$. The statement holds when $m=1$, where $A_1A_2\cdots A_n=\pm 1$. Clearly, $A_1=\pm A_2\cdots A_n$, and we can take $A_2\cdots A_n$ as free variables. We use the first property in Lemma~\ref{lemma:commutatorprop} when there is a need to change the order of two variables, which finishes in $\mathrm{poly}(n)$ steps using a sorting algorithm. Now assume the statement holds for $m=k$. When $m=k+1$, without loss of generality, suppose the first constraint is the added constraint with $A_1A_2\cdots A_l=\pm 1$. Thus, $A_1=\pm A_2\cdots A_l$. Substitute this expression into the other equations and simplify them using Lemma~\ref{lemma:commutatorprop}, which finishes in $\mathrm{poly}(n,m)$ steps using Gaussian elimination. Then, we get a set of $k$ constraints. According to the induction hypothesis, every variable $A_i$ can be represented as a product of free variables up to a sign. So we can plug the free variables into the added constraint $A_1=\pm A_2\cdots A_l$ and get the expression for $A_1$ in terms of free variables. Therefore, the statement holds for $m=k+1$.

\end{proof}

Now, we get the expression for each variable in terms of a set of free variables, resulting in a linear BCS over $A_i$'s and additional sign variables $C_i$'s. 

We further consider the conditions in the original linear BCS and eliminate all the non-free variable $A_i$'s. Using Lemma~\ref{lemma:solvelineareqs}, we find a set of free variables and use them to express all the other variables. Then, we substitute the expressions into the original BCS and obtain a new BCS containing the free variables and $C_i$'s. By the definition of the free variable, every free variable occurs for an even number of times in each constraint of the new BCS by this step, otherwise, it is determined by the other variables through their commutators. By further applying Lemma~\ref{lemma:commutatorprop}, we can get rid of all the $A_i$ variables and obtain a set of equations of $C_i$'s and $C_{kl}$'s where $k,l$ are the commutators of the free variables $A_k$ and $A_l$.
In addition, if $A_i$ and $A_j$ appear in the same constraint of the original BCS, they commute with each other, i.e., $A_iA_jA_iA_j=1$. For each pair of $A_i$ and $A_j$, by applying Lemma~\ref{lemma:solvelineareqs}, replace them in the equation $A_iA_jA_iA_j=1$ via their expressions in terms of the free variables and obtain the other equations of $C_{kl}$'s where $k,l$ are indices of free variables. In the end, we convert the original linear BCS to an equivalent set of equations of $C_i$'s ($i\in[n]$) and $C_{kl}$'s ($k,l$ are indices of free variables), which is just a system of linear equations over $\mathbb{Z}_2$. Note that all the procedures are simply substitutions and the order swaps between variables, which finish in $\mathrm{poly}(n,m)$ steps. Should the original BCS have a Pauli-string solution, we can efficiently solve the newly derived linear equations and get a set of valid values for $C_i$'s and $C_{kl}$'s. 

Note that by this step, we have not finished solving the original BCS over the Pauli group, as we have not determined the operator values of $A_i$'s. The following lemma gives a systematic method to assign legitimate Pauli-string values for all the free variables and, hence, all the variables in the original BCS.

\begin{lemma}\label{lemma:pauliassignment}
For any given set of sign variables $\{C_{ij}=\pm1\}_{1\leq i<j\leq n}$, there exists a set of Pauli strings $\{A_i\}_{1\leq i\leq n}$, such that for any $1\leq i<j\leq n$, $A_iA_jA_iA_j=C_{ij}$. That is, the commutator between $A_i$ and $A_j$ is $C_{ij}$.
\end{lemma}

\begin{proof}
We give an explicit construction. Suppose there are $p$ sign variables equal to $-1$, given by $C_{i_1j_1},C_{i_2j_2},\cdots,C_{i_pj_p}$. Then, we can construct Pauli strings over $p$ qubits according to the following rule: for every $q$'th qubit in each Pauli string, where $1\leq q\leq p$, assign $\sigma_x$ 
for $A_{i_q}$ and $\sigma_z$ for $A_{j_q}$; assign all the other qubits as $\mathbb{I}$. That is,
\begin{equation}
\begin{split}
\text{the $q$'th qubit of $A_k$}=\begin{cases}
\sigma_x, & \text{if $k=i_q$}, \\
\sigma_z, & \text{if $k=j_q$}, \\
\mathbb{I}, & \text{otherwise}.
\end{cases}
\end{split}
\end{equation}
It can be directly checked that this construction satisfies the requirements.

\end{proof}

Later, we take the Mermin-Peres magic square BCS as an example to exhibit the entire procedure. As a side note, the correspondence between traceless symmetric matrices over $\mathbb{Z}_2^{n\times n}$ and Pauli strings was implicitly used in Lemma~7 in Ref.~\cite{ji2013binary}.

Now we summarize the results for determining the Pauli-string solution to a linear BCS. 
\begin{theorem}
For a linear BCS with $n$ variables and $m$ constraints, there exists a classical algorithm that determines whether it has a perfect quantum strategy on the Pauli group in $\mathrm{poly}(n,m)$ steps.
\label{thm:EfficientPauli}
\end{theorem}

The procedures can be summarized as follows:
\begin{enumerate}
\item Solve the BCS as if it is a classical one, with a recording of the commutator and sign changes in each step. Get an expression for each variable in terms of a set of free variables.
\item Substitute the expressions into the original BCS. Get a system of linear functions of $C_i$'s and $C_{kl}$'s and solve them over $\mathbb{Z}_2$.
\item Assign a Pauli string to every free variable. Then derive the operator values of all the variables according to the expressions in step 1.
\end{enumerate}

On the other hand, if there is not a Pauli-string solution to the linear BCS, we shall come to a contradiction somewhere in the procedures.

\begin{corollary}\label{corollary:substitution}
Suppose a linear BCS does not have a satisfying assignment with Pauli-string observables. On the one hand, we can use the substitution method of solving a BCS and obtain a relation for a subset of variables,
\begin{equation}\label{eq:variablestring}
A_{t_1}A_{t_2}\cdots A_{t_k}=-\mathbb{I}.
\end{equation}
On the other hand, by posing the commutation properties of Pauli strings to the BCS variables, the left-hand side can be eliminated to $\mathbb{I}$, resulting in a contradiction.
\end{corollary}

The proof of this corollary is similar to Lemma~\ref{lemma:solvelineareqs}. In brief, in the algorithm for finding Pauli-string solutions, we only use two operations throughout the process: (1) substitution of expressions and (2) swapping two variables in an expression according to Lemma~\ref{lemma:commutatorprop}. If the algorithm cannot find a Pauli-string solution, it must result in a contradiction. As we keep the right-hand side of each formula to be $\pm\mathbb{I}$, the contradiction is thus the form of the statement in the corollary.

As an example, we apply the algorithm to find a Pauli-string solution to the Mermin-Peres magic square BCS. The original BCS is given by
\begin{equation}
\begin{split}
A_1A_2A_3 &= 1, \\
A_4A_5A_6 &= 1, \\
A_7A_8A_9 &= 1, \\
A_1A_4A_7 &= 1, \\
A_2A_5A_8 &= 1, \\
A_3A_6A_9 &= -1.
\end{split}
\end{equation}
After the first step, we find that the BCS variables can be determined by a set of free variables $\{A_5,A_6,A_8,A_9\}$:
\begin{equation}
\begin{split}
A_1 &= C_1A_5A_6A_8A_9, \\
A_2 &= C_2A_5A_8, \\
A_3 &= C_3A_6A_9, \\
A_4 &= C_4A_5A_6, \\
A_7 &= C_7A_8A_9.
\end{split}
\end{equation}
Substituting these expressions into the original BCS, we obtain the set of equations
\begin{equation}\label{eq:MPSubstitute}
\begin{split}
A_1A_2A_3 &= C_1C_2C_3C_{59}C_{56}C_{58}C_{69}C_{89} = 1, \\
A_4A_5A_6 &= C_4C_{56} = 1, \\
A_7A_8A_9 &= C_7C_{89} = 1, \\
A_1A_4A_7 &= C_1C_4C_7C_{59}C_{68}C_{56}C_{58}C_{69}C_{89} = 1, \\
A_2A_5A_8 &= C_2C_{58} = 1, \\
A_3A_6A_9 &= C_3C_{69} = -1.
\end{split}
\end{equation}
Using the commutation conditions between variables in the same constraint as the original BCS, we have the equations
\begin{equation}\label{eq:MPCommute}
\begin{split}
A_1A_2A_1A_2 &= C_{69} = 1, \\
\cdots & \\
A_4A_7A_4A_7 &= C_{59}C_{68}C_{56}C_{58}C_{69}C_{89} = 1, \\
\cdots & \\
A_8A_9A_8A_9 &= C_{89} = 1.
\end{split}
\end{equation}
Solving Eqs.~\eqref{eq:MPSubstitute} and \eqref{eq:MPCommute} over $\mathbb{Z}_2$, we have
\begin{equation}
\begin{split}
C_1&=C_2=C_4=C_7=1, \\
C_3&=-1, \\
C_{56}&=C_{58}=C_{69}=C_{89}=1, \\ C_{59}&=C_{68}=-1.
\end{split}
\end{equation}
With respect to the commutators among the free variables, we obtain two commutators that are equal to $-1$. Assign the free variables as two-qubit Pauli strings,
\begin{equation}
\begin{split}
A_5 &= \sigma_x\otimes\mathbb{I}, \\
A_6 &= \mathbb{I}\otimes\sigma_x, \\
A_8 &= \mathbb{I}\otimes\sigma_z, \\
A_9 &= \sigma_z\otimes\mathbb{I},
\end{split}
\end{equation}
and the other variables are then determined as
\begin{equation}
\begin{split}
A_1 &= -\sigma_y\otimes\sigma_y, \\
A_2 &= \sigma_x\otimes\sigma_z, \\
A_3 &= -\sigma_z\otimes\sigma_x, \\
A_4 &= \sigma_x\otimes\sigma_x, \\
A_7 &= \sigma_z\otimes\sigma_z. \\
\end{split}
\end{equation}
This solution is equivalent to the solution in Eq.~\eqref{eq:MPSolution} in the sense of a unitary transformation, or a relabelling of the variables.

\bibliographystyle{apsrev4-2}

\bibliography{bibBCS}


\end{document}